\documentclass[draftcls,12pt,onecolumn,oneside,compress]{IEEEtran}

\usepackage{amsmath,amssymb,graphicx,cite,color,enumerate,amsthm,algorithm,algorithmic}

\theoremstyle{definition}
\newtheorem{rem}{Remark}
\newtheorem{theorem}{Theorem}
\newtheorem{corollary}{Corollary}
\newtheorem{definition}{Definition}
\begin{document}
\title{Semi-coherent Detection and Performance Analysis for Ambient Backscatter System}

\author{Jing Qian, Feifei Gao, Gongpu Wang, Shi Jin, and Hongbo Zhu

\thanks{J. Qian and F. Gao are with Tsinghua National Laboratory for Information Science and Technology (TNList), Beijing 100084, P. R. China (Email: qian-j13@mails.tsinghua.edu.cn, feifeigao@ieee.org).
G. Wang is with School of Computer and Information Technology, Beijing Jiaotong University, Beijing 100044, P. R. China (Email: gpwang@bjtu.edu.cn).
S. Jin is with the National Communications Research Laboratory, Southeast University, Nanjing 210096, P. R. China (Email:
jinshi@seu.edu.cn).
H. Zhu is with the Jiangsu Key Laboratory of Wireless Communications, Nanjing University of Posts and Telecommunications, Nanjing 210003, P. R. China (Email: zhuhb@njupt.edu.cn).}
}

\maketitle
\thispagestyle{empty}
\vspace{-10mm}

\begin{abstract}
We study a novel communication mechanism, ambient backscatter, that utilizes radio frequency (RF) signals transmitted from an ambient source as both energy supply and information carrier to enable communications between low-power devices. Different from existing non-coherent schemes, we here design the semi-coherent detection, where channel parameters can be obtained from unknown data symbols and a few pilot symbols.  We first derive the optimal detector for the complex Gaussian ambient RF signal from likelihood ratio test  and compute the corresponding closed-form bit error rate (BER). To release the requirement for prior knowledge of the ambient RF signal, we next design a suboptimal energy detector with ambient RF signals being either the complex Gaussian or the phase shift keying (PSK). The corresponding detection thresholds, the analytical BER, and the outage probability are also obtained in closed-form. Interestingly, the complex Gaussian source would cause an error floor   while the PSK  source does not, which brings nontrivial indication of constellation design as opposed to the popular Gaussian-embedded  literatures.  Simulations are   provided to corroborate the theoretical studies.

\end{abstract}


\IEEEpeerreviewmaketitle

\section{Introduction}
The Internet of Things (IoT) \cite{refs:tiotas,refs:tiotas1} describes the next generation of Internet, where all things could be accessed and identified through the Internet via sensing devices \cite{refs:btioturtreeaoopr,refs:wesfiotapa}.  As emerging wirelessly sensory technologies have significantly improved the capability of devices, IoT is being extended to ambient intelligence and autonomous control \cite{refs:iwtsiotdqsaopows,refs:iipaaomrtcm,refs:tioffrttnpns}.  Such an extension,  however, also leads to a key bottleneck in its development: since such a huge number of devices need to be battery-free and has to be powered with harvested energies, generating radio waves themselves typically seems to be unrealistic.

One solution is the backscatter communication \cite{refs:bcarceama,refs:ebbcisehe}, where devices can transmit their data through modulating and reflecting incident radio frequency (RF) signals. It is distinct from traditional wireless communications in that backscatter devices consume power orders-of-magnitude less, as they require no energy hungry components such as oscillators. A typical application example is the radio frequency identification (RFID) consisting of an active reader (the transceiver) and a passive tag (the backscatter node). Specifically, the reader can generate continuous carrier waves, while the tag modulates its information onto the carrier wave by adapting its antenna impedance loading to vary the reflection coefficient and then backscatters the signal to the reader.

In order to enable ubiquitous communications between battery-free devices, a novel communication mechanism, called ambient backscatter, was introduced in \cite{refs:abwcoota}, which leverages existing ambient RF signals and applies them into the backscatter communication. The ambient backscatter differs from conventional backscatter communications in that it does not require a centralized high-cost infrastructure (e.g., a RFID reader) to transmit pre-requisite signals and to initiate/control communications with devices.
Moreover, since ambient RF signals are always available, it enables the communication between passive devices almost everywhere and anytime.

Following \cite{refs:abwcoota}, the way of connecting ambient backscatter tags with the Internet via the existing Wi-Fi infrastructure was designed in \cite{refs:wbicffd}. In \cite{refs:tabc}, the authors presented the multi-antenna interference cancellation scheme operating on the backscatter devices. Nevertheless, these works mainly focus on the hardware design and the prototype presentation with modest decoding performance but did not provide the fundamental results from theoretical aspects.

Some exploration about signal detection for the ambient backscatter communication was presented in \cite{refs:udabafabcs,refs:abcsdapa,refs:sdabafrpduab}, where the tag tends to employ the on-off signaling with a low data rate, and the reader can decode tag's information by simple detection strategies. Another transmission model was proposed in \cite{refs:sdfabswmra}, where the reader is equipped with multiple antennas. The authors of \cite{refs:sdaoabswdm} looked into the non-coherent symbol detection under the condition that the channel state information is unknown, and provided a method to estimate the system parameters without sending pilots. Meanwhile, a detection algorithm based on statistical covariances is suggested in \cite{refs:scbsdfabcs}, which requires extremely large number of samples.

In this paper, we provide a fundamental study over the semi-coherent detection of the classical three-node ambient backscatter system\footnote{Some of our  preliminary results  were published in \cite{refs:sdapaotabs}.}, where the channel state information (CSI) is unknown and training symbols are sent to acquire the detection-required parameters rather than the channels themselves. We first derive the optimal detector from the likelihood-ratio test of the received signal vector with the assumption of complex Gaussian ambient RF signals. As the optimal detector requires the availability of the prior knowledge of ambient RF signals and comes with a less informative BER expression, a suboptimal energy detector is designed, where we consider both the complex Gaussian and the phase shift keying (PSK) ambient RF signals, and derive their corresponding optimal detection thresholds. The analytical bit error rate (BER) as well as the BER-based outage probability are obtained in closed-form, which tells more insight of the system parameters and helps choosing the optimal parameters. Interestingly, we demonstrate that the BER with complex Gaussian ambient RF signals would exhibit an error floor while that with PSK ambient RF signals does not.  A practical approach that estimates the parameters from the unknown data symbols and a few  pilot symbols is also proposed. Finally simulation results demonstrate the effectiveness of different detectors as well as the correctness of the theoretical analysis.

The rest of the paper is organized as follows. Section II outlines the system model. In Section III, the optimal detector and the suboptimal energy detector is derived, along with their corresponding performance analysis. In Section IV, the parameter estimation for the semi-coherent detection is proposed. The simulation results are provided in Section V and Section VI concludes the paper.

\textbf{Notations:} Vectors and matrices are boldfaced letters: the Hermitian, the inverse, and the determinant of matrix ${\boldsymbol A}$ are denoted by ${\boldsymbol A}^H$, ${\boldsymbol A}^{-1}$, and ${\rm det}({\boldsymbol A})$, respectively; ${\boldsymbol 1}_N$ and ${\boldsymbol I}_N$ are the $N$-order unit vector and the $N$-order unit matrix, respectively; $\left\|\boldsymbol y\right\|$ denotes the Euclidean norm of vector $\boldsymbol y$. Scalars are lowercase letters: $h^*$, $|h|$, and $\Re\{h\}$ denotes the conjugate, the modulus, and the real part of complex number $h$, respectively. $\mathbb{E}\{X\}$ and $\mathrm{var}\{X\}$ are the statistical expectation and the statistical variance of random variable $X$, respectively; $\mathcal{N}(\mu,\sigma^2)$ and $\mathcal{CN}(\mu,\sigma^2)$ respectively denotes the Gaussian distribution and the circularly symmetric complex Gaussian distribution with mean $\mu$ and variance $\sigma^2$.

\section{System Model}
\label{sec:system model}

\begin{figure}[t]
\centering
\includegraphics[width=100mm]{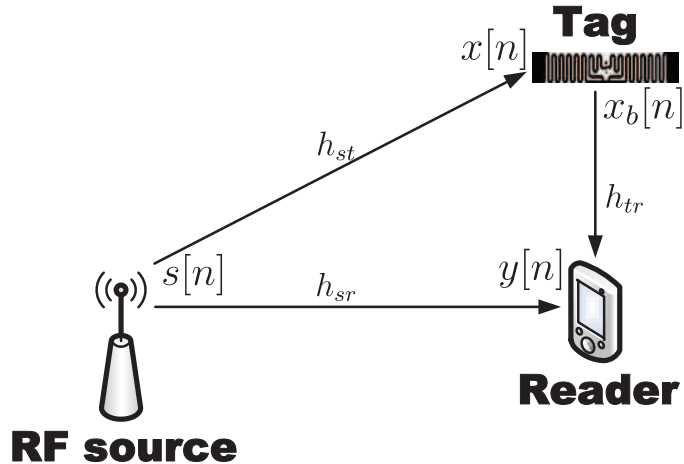}
\caption{A three-node ambient backscatter system consisting of the RF signal source, a passive tag and a reader.}
\label{fig:system model}
\end{figure}

Consider a classical three-node ambient backscatter system as depicted in Fig. \ref{fig:system model}. Denote $h_{st}$, $h_{sr}$, and $h_{tr}$ as the coefficients of the channels from the source to the tag, from the source to the reader, and from the tag to the reader, respectively. A frequency-flat and block-fading channel model is assumed, where all the channels are constant within the channel coherence time but  may vary independently in different coherence intervals.

The signal from the ambient RF source can be received by both the tag and the reader. The tag transfers its binary symbols to the reader by choosing whether to backscatter the incident RF signal or not. Specifically, if the tag wants to transmit the symbol ``0'', it will adjust its impedance so that little of the incident signal can be reflected; while if it wants to transmit the symbol ``1'', some of the incident signal will be backscattered to the reader. The reader then senses the changes in the received signals and thus decode the transmitted symbols of the tag.

Mathematically, the signal received by the tag can be expressed as
\begin{align}
x[n]= h_{st}s[n],
\end{align}
where $s[n]$ is the unknown ambient RF signal. Since the tag only consists of passive components related to backscattering and involves little signal processing operation, the thermal noise at the tag could be negligible \cite{refs:aattbfurt}.

Suppose the transmitted binary symbols of the tag is $d[n]\in\{0,1\}$, where ``0'' and ``1'' are of equal transmit probabilities. The signal backscattered by the tag is
\begin{align}
x_b[n]= \alpha d[n]x[n],
\end{align}
where the real number $\alpha$ is the tag coefficient related to scattering efficiency and antenna gain.

The reader receives the superposition of the signal from the RF source and the signal backscattered from the tag:
\begin{align}\label{eq:y}
y[n]&= h_{sr}s[n] + h_{tr}x_b[n]+ w[n]= (h_{sr} +\alpha h_{st}h_{tr}d[n])s[n] + w[n],
\end{align}
where $w[n]$ is the zero-mean additive white Gaussian noise (AWGN) with variance $N_w$, i.e., $w[n]\sim\mathcal{CN}(0,N_w)$.

Compared with the conventional  communications model, (\ref{eq:y}) is more challenging in that, besides the detected symbol $d[n]$, $h_{st}$, $h_{sr}$, $h_{tr}$, $\alpha$, $s[n]$ and $w[n]$ are all unknown to the reader, while these parameters are coupled with each other in a more complicated way.

\section{Symbol Detection}
Different from the high-speed data transmission in conventional wireless networks, the communication involved in the ambient backscatter system is generally in a low-rate manner. For example, the long-term parameters feedback in sensor networks or in the IoT. Thus, the tag will transmit at a much lower rate than the rate of the ambient RF signal, say, $d[n]$ remains unchanged for $N$ (an even number without loss of generality) consecutive $s[n]$'s.

For clarity, let us omit the index $n$ in $d[n]$ and use $d$ to denote one symbol  of the tag. Meanwhile, denote $\boldsymbol{y}=[y[1],\cdots,y[N]]^T$ as its corresponding received signal vector at the reader, where
\begin{align}\label{eq:y}
y[n]=\left\{\begin{array}{ll}
h_0s[n]+w[n],&~~~~d=0,\\
h_1s[n]+w[n],&~~~~d=1,
\end{array}\right.
\end{align}
and we define $h_0=h_{sr}$ and $h_1 = h_{sr} + \alpha h_{st}h_{tr}$ for notation simplicity.

\subsection{Optimal Detector with the Complex Gaussian Ambient Source}
In this section, we assume that the ambient RF signal follows the complex Gaussian distribution, i.e., $s[n]\sim\mathcal{CN}(0,P_s)$.

Denote $\mathcal{H}_0$ and $\mathcal{H}_1$ as the hypotheses that the tag's transmitted symbol is   $d=0$ and $d=1$, respectively. The received signal vector $\boldsymbol{y}$ is then a complex Gaussian vector with
\begin{align}\label{eq:ydistri}
\boldsymbol{y}\sim\left\{\begin{array}{ll}
\mathcal{CN}(\boldsymbol{0},\sigma_0^2\boldsymbol{I}_N),&~~~~\mathcal{H}_0,\\
\mathcal{CN}(\boldsymbol{0},\sigma_1^2\boldsymbol{I}_N),&~~~~\mathcal{H}_1,
\end{array}\right.
\end{align}
where
\begin{align}\label{eq:sigma2}
\sigma_0^2\triangleq|h_0|^2P_s+N_w,~~~~~~\sigma_1^2\triangleq|h_1|^2P_s+N_w.
\end{align}

\begin{rem}
Although the knowledge of CSI is unavailable, the values of $\sigma_i^2$ can be estimated in a way as will be presented in Section \ref{sec:estimation} and
are assumed known throughout our discussions. Moreover,
estimating $\sigma_i^2$ is more robust than estimating the channels themselves since the channel energy (or equivalently the channel amplitude) varies much slower than the instantaneous CSI.
\end{rem}

Under the maximum likelihood paradigm \cite{refs:dc}, the optimal symbol detection can be achieved from the likelihood ratio testing, defined as
\begin{align}
\Lambda(\boldsymbol{y}) =\frac{p\left(\boldsymbol{y}|\mathcal{H}_0\right)}{p\left(\boldsymbol{y}|\mathcal{H}_1\right)} =\left(\frac{\sigma_1^2}{\sigma_0^2}\right)^N\exp\left(\frac{\sigma_0^2-\sigma_1^2}{\sigma_0^2\sigma_1^2}Z\right),
\end{align}
where $Z=\|\boldsymbol{y}\|^2$, and $p(\boldsymbol{y}|\mathcal{H}_i)$ represents the probability density function (PDF) of $\boldsymbol{y}$ under the hypothesis $\mathcal{H}_i$. Obviously, the likelihood ratio depends only on $Z$, i.e., the energy of the received signal vector, which is the key statistics of the testing.

However, different from conventional detection methods, whether $\Lambda(\boldsymbol{y})$ is increasing over $Z$ or not depends on the relationship between the values of $\sigma_0^2$ and $\sigma_1^2$. Thus, the decision rule could be made through
\begin{align}\label{eq:thcond}
\Lambda(\boldsymbol{y})~\mathop{\gtrless}^{\mathcal{H}_0}_{\mathcal{H}_1}~1 ~~\Longleftrightarrow~~
\left\{\begin{array}{ll}
\displaystyle Z~\mathop{\gtrless}^{\mathcal{H}_0}_{\mathcal{H}_1}~T_{h}^{\mathrm{CG-op}},&~~~\sigma_0^2>\sigma_1^2,\\[2 mm]
\displaystyle Z~\mathop{\lessgtr}^{\mathcal{H}_0}_{\mathcal{H}_1}~T_{h}^{\mathrm{CG-op}},&~~~\sigma_0^2<\sigma_1^2,
\end{array}\right.
\end{align}
where $T_h^{\mathrm{CG-op}}$ is the threshold for locating the range of the energy $Z$. In fact, (\ref{eq:thcond}) can be referred to
as a modified energy detection.

\begin{rem}
If $\sigma_0^2=\sigma_1^2$, then the two hypotheses cannot be discriminated and the detection fails. Nevertheless, the probability for such scenario to happen is nearly zero.
\end{rem}

\begin{theorem}
The threshold for the optimal ML detector can be expressed as
\begin{align}\label{eq:optimalth}
T_{h}^{\mathrm{CG-op}} =\frac{N\sigma_0^2\sigma_1^2}{\sigma_1^2-\sigma_0^2}\ln\frac{\sigma_1^2}{\sigma_0^2}.
\end{align}
\end{theorem}

\begin{proof}
The threshold is obtained from (\ref{eq:thcond}) by solving $\Lambda(\boldsymbol{y})=1$.
\end{proof}

We summarize  the optimal ML detector in Algorithm \ref{alg:optimal detector1}.
\begin{algorithm}[h!]
\caption{Optimal Detector}
\label{alg:optimal detector1}
\begin{algorithmic}[1]
\REQUIRE{The received signal vectors at the reader, $\boldsymbol y$.}
\ENSURE{The detected result of the transmitted symbol of the tag, $\hat{d}$.}
\STATE Calculate the signal energy $Z=\|\boldsymbol y\|^2$;\\
\STATE Obtain the parameters $\sigma_0^2$ and $\sigma_1^2$, and calculate the detection threshold $T_h^{\mathrm{CG-op}}$;\\
\IF {$\sigma_0^2>\sigma_1^2$}
\STATE
{
\textbf{if}
{$Z\geq T_h^{\mathrm{CG-op}}$}~~~
\textbf{then}
$\hat{d}=0$~~~
\textbf{else}
$\hat{d}=1$~~~
\textbf{end if}
}
\ELSE
\STATE
{
\textbf{if}
{$Z\leq T_h^{\mathrm{CG-op}}$}~~~
\textbf{then}
$\hat{d}=0$~~~
\textbf{else}
$\hat{d}=1$~~~
\textbf{end if}
}
\ENDIF
\RETURN $\hat{d}$
\end{algorithmic}
\end{algorithm}

\begin{theorem}
The BER of the optimal ML detector can be expressed as
\begin{align}\label{eq:optimalpb}
P_{b}^{\mathrm{CG-op}} &
=\frac{1}{2\Gamma(N)}\left[\gamma\left(N,\frac{N\sigma_{\min}^2}{\sigma_1^2-\sigma_0^2}\ln\frac{\sigma_1^2}{\sigma_0^2}\right)
+\Gamma\left(N,\frac{N\sigma_{\max}^2}{\sigma_1^2-\sigma_0^2}\ln\frac{\sigma_1^2}{\sigma_0^2}\right)\right],
\end{align}
where $\sigma_{\max}^2=\max\{\sigma_0^2,\sigma_1^2\}$, $\sigma_{\min}^2=\min\{\sigma_0^2,\sigma_1^2\}$, and
\begin{align}
\!\gamma(N,x)=\!\int_0^x \!t^{N-1} \mathrm{e}^{-t}\mathrm{d}t~~~~\textrm{and}~~~~\Gamma(N,x)=\!\int_x^\infty \!t^{N-1} \mathrm{e}^{-t}\mathrm{d}t
\end{align}
denote the lower and the upper incomplete gamma functions, respectively.

\end{theorem}

\begin{proof}
According to (\ref{eq:thcond}), for the case of $\sigma_0^2>\sigma_1^2$, the BER can be derived as
\begin{align}\label{eq:theorypb0}
P_{b}^{\mathrm{CG-op}} &= \Pr(\mathcal{H}_0)\Pr(Z\leq T_{h}^{\mathrm{CG-op}}|\mathcal{H}_0) + \Pr(\mathcal{H}_1)\Pr(Z\geq T_{h}^{\mathrm{CG-op}}|\mathcal{H}_1) \nonumber\\
     &= \frac{1}{2}\int_0^{T_{h}^{\mathrm{CG-op}}} f_Z(z|\mathcal{H}_0) \mathrm{d}z + \frac{1}{2}\int_{T_{h}^{\mathrm{CG-op}}}^\infty f_Z(z|\mathcal{H}_1) \mathrm{d}z,
\end{align}
where $f_Z(z|\mathcal{H}_i)$ is the PDF of $Z$ under the hypothesis $\mathcal{H}_i$.

It can be readily  known that $Z$ is a central chi-square random variable with $2N$ degrees of freedom (DOF). Then, there is  \cite{refs:pdigrv}
\begin{align}
f_Z(z|\mathcal{H}_i)=\frac{z^{N-1}\mathrm{e}^{-\frac{z}{\sigma_i^2}}}{\Gamma(N){\sigma_i}^{2N}},~~~~~~~i=0,1,
\end{align}
where $\Gamma(\cdot)$ denotes the gamma function. Then the BER (\ref{eq:theorypb0}) is further derived as
\begin{align}\label{eq:theorypb1}
P_{b}^{\mathrm{CG-op}} &= \frac{1}{2\Gamma(N)}\left[\gamma\left(N,\frac{T_{h}^{\mathrm{CG-op}}}{\sigma_0^2}\right)+\Gamma\left(N,\frac{T_{h}^{\mathrm{CG-op}}}{\sigma_1^2}\right)\right].
\end{align}

Similarly, for the case of $\sigma_0^2<\sigma_1^2$, the corresponding BER is  obtained as
\begin{align}\label{eq:theorypb2}
P_{b}^{\mathrm{CG-op}} &= \frac{1}{2\Gamma(N)}\left[\Gamma\left(N,\frac{T_{h}^{\mathrm{CG-op}}}{\sigma_0^2}\right)+\gamma\left(N,\frac{T_{h}^{\mathrm{CG-op}}}{\sigma_1^2}\right)\right].
\end{align}
Moreover, (\ref{eq:theorypb1}) and (\ref{eq:theorypb2}) can be integrated into one, and thus we obtain (\ref{eq:optimalpb}).
\end{proof}

For relatively large $N$, there are approximations \cite{refs:homfwfgamt}:
\begin{align}\label{eq:approximation}
\frac{\gamma(N,x)}{\Gamma(N)}\approx 1-Q(x_1),~~~~~~\frac{\Gamma(N,x)}{\Gamma(N)}\approx Q(x_1),~~~~~~\textrm{and}~~~~~~x_1=\frac{x}{\sqrt{N}}-\sqrt{N}.
\end{align}
We can further approximate the expression in (\ref{eq:optimalpb}) as
\begin{align}\label{eq:approtheoreticalpb1}
P_{b}^{\mathrm{CG-op}} &\approx\frac{1}{2}Q\left(\sqrt{N}-\frac{\sqrt{N}\sigma_{\min}^2}{\sigma_0^2-\sigma_1^2}\ln\frac{\sigma_0^2}{\sigma_1^2}\right)+\frac{1}{2}Q\left(\frac{\sqrt{N}\sigma_{\max}^2}{\sigma_0^2-\sigma_1^2}\ln\frac{\sigma_0^2}{\sigma_1^2}-\sqrt{N}\right),
\end{align}
which indicates  that the difference between $\sigma_0^2$ and $\sigma_1^2$ may be a crucial factor to the detection performance.

\begin{rem}
The optimal detector may not obtain the same error probability for $\mathcal{H}_0$ and $\mathcal{H}_1$, i.e., $\Pr(\hat{d}=1|\mathcal{H}_0)\neq\Pr(\hat{d}=0|\mathcal{H}_1)$, which is generally referred as the unbalanced BER \cite{refs:obcfwedouos}.  In some case, a balanced BER detector\footnote{Balanced BER means that there is not any distinction introduced by the detection method to the status of different bits, and thus the BER performance does not rely on the detection method.} is required for. Referring to (\ref{eq:theorypb1}) and (\ref{eq:theorypb2}), the balanced BER detector with its threshold $T_h^{\mathrm{ba}}$ can be achieved from
\begin{align}\label{eq:equthresh}
\gamma\left(N,\frac{T_h^{\mathrm{ba}}}{\sigma_{\max}^2}\right)=\Gamma\left(N,\frac{T_h^{\mathrm{ba}}}{\sigma_{\min}^2}\right),
\end{align}
where it is difficult to get the exact solution of $T_h^{\mathrm{ba}}$. However, with the approximation in (\ref{eq:approximation}), we can further rewrite  (\ref{eq:equthresh}) as
\begin{align}\label{eq:appequthresh}
Q\left(\sqrt{N}-\frac{T_h^{\mathrm{ba}}}{\sqrt{N}\sigma_{\max}^2}\right)=Q\left(\frac{T_h^{\mathrm{ba}}}{\sqrt{N}\sigma_{\min}^2}-\sqrt{N}\right),
\end{align}
and obtain the threshold for the balanced BER detector
\begin{align}\label{eq:equalthresh}
T_h^{\mathrm{ba}}=\frac{2N\sigma_0^2\sigma_1^2}{\sigma_0^2+\sigma_1^2}.
\end{align}
\end{rem}

\subsection{Suboptimal Detector with the Complex Gaussian Ambient Source}
From (\ref{eq:optimalpb}) or (\ref{eq:approtheoreticalpb1}), we cannot obtain a clear clue about how the system parameters will affect the detection performance. Thus, we here design a suboptimal detector which does not gain any undesirable performance loss, but requires less prior knowledge of the ambient RF signal and yields a simpler and more informative BER expression.

From (\ref{eq:thcond}), we know the energy of the received signal vector  $Z$ is the key statistics of the detection, and the energy detection with a proper threshold could be the optimal detection. Thus, the decision metric can be switched from the PDF of $\boldsymbol{y}$ to PDF of $Z$. From another perspective, the energy $Z=\sum\limits_{n=1}^N|y[n]|^2$ can also be regarded as the sum of $N$ independent 2-DOF central chi-square random variables with the identical mean $\sigma_i^2$ and variance $\sigma_i^4$ under the hypothesis $\mathcal{H}_i$. When $N$ is relatively large\footnote{ Normally, $N = 30$ is adequate for most applications. However, if the PDF of $|y[n]|^2$ is smooth, then the value of $N$ as low as 5 can be used \cite{refs:prvasp}.}, $Z$ asymptotically becomes a Gaussian random variable from the central limit theorem \cite{refs:acovrdm}.
Then the distribution of $Z$ under hypothesis $\mathcal{H}_i$ can be approximated as $Z|\mathcal{H}_i\sim\mathcal{N}(\mu^{\mathrm{CG}}_i,\varsigma^{\mathrm{CG}}_i)$ with PDF
\begin{align}
\tilde{f}_{Z}(z|\mathcal{H}_i)=\frac{1}{\sqrt{2\pi\varsigma_i^{\mathrm{CG}}}}\exp\left[-\frac{\left(z-\mu^{\mathrm{CG}}_i\right)^2}{2\varsigma_i^{\mathrm{CG}}}\right],~~~~~~~~i=0,1,
\end{align}
where
\begin{align}\label{eq:hatmu}
\mu^{\mathrm{CG}}_i=N\sigma_i^2,~~~~~~~~\varsigma^{\mathrm{CG}}_i=N\sigma_i^4,~~~~~~~~i=0,1,
\end{align}
are the means and the variances of $Z$ under the hypothesis $\mathcal{H}_i$, respectively.

The detection rule for the suboptimal detector is reformulated as
\begin{align}\label{eq:hatthcond}
\tilde{f}_{Z}(z|\mathcal{H}_0)\mathop{\gtrless}^{\mathcal{H}_0}_{\mathcal{H}_1}\tilde{f}_{Z}(z|\mathcal{H}_1) \Longleftrightarrow
\left\{\begin{array}{ll}
\displaystyle {Z}\mathop{\gtrless}^{\mathcal{H}_0}_{\mathcal{H}_1}T_{h}^{\mathrm{CG-sub}},&~~~~\sigma_0^2>\sigma_1^2,\\[2 mm]
\displaystyle {Z}\mathop{\lessgtr}^{\mathcal{H}_0}_{\mathcal{H}_1}T_{h}^{\mathrm{CG-sub}},&~~~~\sigma_0^2<\sigma_1^2.
\end{array}\right.
\end{align}
Namely, the suboptimal detector is also a type of energy detection but with a different threshold from the optimal one (\ref{eq:optimalth}).

\subsubsection{General Case}
We first present the general case of the suboptimal detection.
\begin{theorem}
The threshold for the suboptimal detector can be expressed as
\begin{align}\label{eq:suboptimalth}
T_h^{\mathrm{CG-sub}}  = \frac{N\sigma_0^2\sigma_1^2}{\sigma_0^2+\sigma_1^2}\left[1+\sqrt{1+\frac{2(\sigma_0^2+\sigma_1^2)}{N(\sigma_1^2-\sigma_0^2)}\ln\frac{\sigma_1^2}{\sigma_0^2}}\right].
\end{align}
\end{theorem}

\begin{proof}
The threshold $T_{h}^{\mathrm{CG-sub}}$ for the suboptimal detector can be computed from
\begin{align}\label{eq:mldetector}
\tilde{f}_{Z}(T_{h}^{\mathrm{CG-sub}}|\mathcal{H}_0)=\tilde{f}_{Z}(T_{h}^{\mathrm{CG-sub}}|\mathcal{H}_1).
\end{align}
Taking the natural logarithm of both sides of (\ref{eq:mldetector}) and rearranging the terms, we obtain
\begin{align}\label{eq:polynomial}
c_1 {(T_{h}^{\mathrm{CG-sub}}})^2+c_2T_{h}^{\mathrm{CG-sub}}+c_3=0,
\end{align}
where
\begin{align}
&c_1=\varsigma^{\mathrm{CG}}_1-\varsigma^{\mathrm{CG}}_0,~~~~~~~~~~c_2=2(\mu^{\mathrm{CG}}_1\varsigma^{\mathrm{CG}}_0-\mu^{\mathrm{CG}}_0\varsigma^{\mathrm{CG}}_1),\\
&c_3=(\mu^{\mathrm{CG}}_0)^2\varsigma^{\mathrm{CG}}_1-(\mu^{\mathrm{CG}}_1)^2\varsigma^{\mathrm{CG}}_0-\varsigma^{\mathrm{CG}}_0\varsigma^{\mathrm{CG}}_1\ln\frac{\varsigma^{\mathrm{CG}}_1}{\varsigma^{\mathrm{CG}}_0}.
\end{align}

As $T_{h}^{\mathrm{CG-sub}}$ is the detection threshold of the received signal energy, only the positive root of (\ref{eq:polynomial}) is valid, which gives the threshold (\ref{eq:suboptimalth}).
\end{proof}

We next demonstrate the BER performance of the suboptimal detector, which tells more insight of the performance-affected parameters and would help design the system parameters.
\begin{theorem}\label{theorem:cgsubber}
The BER for the suboptimal detector can be expressed as
\begin{align}\label{eq:gaussianpb}
P_{b}^{\mathrm{CG-sub}}=\frac{1}{2}-\frac{1}{2}Q\left(\frac{T_{h}^{\mathrm{CG-sub}}-N\sigma^2_{\max}}{\sqrt{N}\sigma_{\max}^2}\right)+\frac{1}{2}Q\left(\frac{T_{h}^{\mathrm{CG-sub}}-N\sigma^2_{\min}}{\sqrt{N}\sigma_{\min}^2}\right).
\end{align}
\end{theorem}
\begin{proof}
According to (\ref{eq:hatthcond}), if $\sigma_0^2>\sigma_1^2$, the corresponding BER is
\begin{align}\label{eq:gaussianpb1}
P_{b}^{\mathrm{CG-sub}}  &= \Pr(\mathcal{H}_0)\Pr(Z<T_{h}^{\mathrm{CG-sub}}|\mathcal{H}_0) + \Pr(\mathcal{H}_1)\Pr(Z>T_{h}^{\mathrm{CG-sub}}|\mathcal{H}_1) \nonumber\\
    &= \frac{1}{2}\int_{-\infty}^{T_{h}^{\mathrm{CG-sub}}} \tilde{f}_{Z}(z|\mathcal{H}_0) \mathrm{d}z + \frac{1}{2}\int_{T_{h}^{\mathrm{CG-sub}}}^\infty \tilde{f}_{Z}(z|\mathcal{H}_1) \mathrm{d}z  \nonumber\\
    &= \frac{1}{2}-\frac{1}{2}Q\left(\frac{T_{h}^{\mathrm{CG-sub}}-\mu^{\mathrm{CG}}_0}{\sqrt{\varsigma^{\mathrm{CG}}_0}}\right)
    +\frac{1}{2}Q\left(\frac{T_{h}^{\mathrm{CG-sub}}-\mu^{\mathrm{CG}}_1}{\sqrt{\varsigma^{\mathrm{CG}}_1}}\right).
\end{align}
If $\sigma_0^2<\sigma_1^2$, the BER is similarly derived as
\begin{align}\label{eq:gaussianpb2}
P_{b}^{\mathrm{CG-sub}}=\frac{1}{2}Q\left(\frac{T_{h}^{\mathrm{CG-sub}}-\mu^{\mathrm{CG}}_0}{\sqrt{\varsigma^{\mathrm{CG}}_0}}\right)
+\frac{1}{2}-\frac{1}{2}Q\left(\frac{T_{h}^{\mathrm{CG-sub}}-\mu^{\mathrm{CG}}_1}{\sqrt{\varsigma^{\mathrm{CG}}_1}}\right).
\end{align}
Therefore, the BER (\ref{eq:gaussianpb}) is obtained by integrating (\ref{eq:gaussianpb1}) and (\ref{eq:gaussianpb2}) into one.
\end{proof}

\subsubsection{Special Case with Large $N$}

We next focus on analyzing the special case with large $N$, where much more results can be obtained.
\begin{corollary}
For a relatively large value of $N$, the asymptotic one of (\ref{eq:suboptimalth}) is expressed as
\begin{align}\label{eq:ta}
\tilde{T}_h^{\mathrm{CG-sub}} \approx\frac{2N\sigma_0^2\sigma_1^2}{\sigma_0^2+\sigma_1^2},
\end{align}
and the asymptotic BER is given by
\begin{align}\label{eq:cgapb}
\tilde{P}_{b}^{\mathrm{CG-sub}}
= Q\left(\frac{\sqrt{N}|\sigma_1^2-\sigma_0^2|}{\sigma_0^2+\sigma_1^2}\right)
= Q\left(\frac{\sqrt{N}\Delta}{\Sigma+2/\gamma}\right),
\end{align}
where
\begin{align}
\gamma=\frac{P_s}{N_w},~~~~\Delta=||h|_0^2-|h_1|^2|,~~~~\Sigma=|h_0|^2+|h_1|^2.
\end{align}
\end{corollary}
\begin{proof}
The result (\ref{eq:cgapb}) is easily obtained by substituting the asymptotic threshold (\ref{eq:ta}) and the expressions of $\sigma_i^2$ (\ref{eq:sigma2}) into (\ref{eq:gaussianpb}). Note that $\gamma$ is the signal-to-noise ratio (SNR) of the ambient RF source.
\end{proof}

It can be readily checked that $\tilde{P}_{b}^{\mathrm{CG-sub}}$ is an decreasing function of $\frac{\sqrt{N}\Delta}{\Sigma+2/\gamma}$, i.e., larger SNR, larger $N$, larger $\Delta$, and smaller $\Sigma$ all conduce to better detection performance. It may differ from the conventional binary detection theory where the performance is mainly controlled by SNR and $N$.

\begin{rem}
Different from \cite{refs:udabafabcs} and the proposed  optimal detector (\ref{eq:optimalth}), the suboptimal detector achieves the same error probability for $d_k=0$ and $d_k=1$ at the threshold (\ref{eq:ta}), i.e.,
\begin{align}
&\Pr(\hat{d}=0|\mathcal{H}_1)-\Pr(\hat{d}=1|\mathcal{H}_0)
=\frac{1}{2}-\frac{1}{2}Q\left(\frac{\tilde{T}_{h}^{\mathrm{CG-sub}}-N\sigma^2_0}{\sqrt{N}\sigma_0^2}\right)-\frac{1}{2}Q\left(\frac{\tilde{T}_{h}^{\mathrm{CG-sub}}-N\sigma^2_1}{\sqrt{N}\sigma_1^2}\right)\nonumber\\
&=\frac{1}{2}\left[1-Q\left(\frac{\sqrt{N}(\sigma_1^2-\sigma_0^2)}{\sigma_0^2+\sigma_1^2}\right)-Q\left(\frac{\sqrt{N}(\sigma_0^2-\sigma_1^2)}{\sigma_0^2+\sigma_1^2}\right)\right]
=0.
\end{align}
Moreover, it is readily seen that $T_{h}^{\mathrm{ba}} =\tilde{T}_h^{\mathrm{CG-sub}}$. The suboptimal detector with large $N$ achieves the same performance as the optimal detector with balanced BER.
\end{rem}

By carefully checking (\ref{eq:cgapb}), we find that there exists an irreducible BER in terms of SNR, i.e., when SNR turns to infinity, the BER does not go to zero but will approach an error floor.

\begin{corollary}
As the SNR goes to infinity, the BER of the suboptimal detector meets an error floor at
\begin{align}\label{eq:pbfloor}
P_b^{\textrm{floor}}= Q\left(\frac{\sqrt{N}\Delta}{\Sigma}\right)
           \approx \frac{1}{12}\mathrm{exp}\left(-\frac{N\Delta^2}{2\Sigma^2}\right)+\frac{1}{4}\mathrm{exp}\left(-\frac{2N\Delta^2}{3\Sigma^2}\right).
\end{align}
\end{corollary}
\begin{proof}
The result is obtained by omitting the term $\frac{2}{\gamma}$ in (\ref{eq:cgapb}) when SNR turns to infinity, and we utilize a simple but accurate approximation of $Q(x)$ \cite{refs:nebaaftcoepifc}
\begin{align}
Q(x)\approx\frac{1}{12}\mathrm{exp}\left({-\frac{x^2}{2}}\right)+\frac{1}{4}\mathrm{exp}\left({-\frac{2x^2}{3}}\right),~~~~~x\geq0.,
\end{align}
for the approximate equality in (\ref{eq:pbfloor}).
\end{proof}

Clearly, the position of the error floor is related to the value of $N$ and ${\Delta}/{\Sigma}$, where the latter reflect the impacts of the channels.
We then define relative channel difference (RCD) as
\begin{align}
\textrm{RCD}\triangleq\frac{\Delta}{\Sigma}=\frac{||h_0|^2-|h_1|^2|}{|h_0|^2+|h_1|^2}.
\end{align}
Since the detection is mainly based on checking the energies under two different channel situations, when  SNR increases to a certain extent the impact of the high SNR on enlarging the energy difference is not dominant, while the relative difference between the two channel situations, i.e., RCD, will play a very important role for the detection performance.

\begin{definition}
Define the outage probability as the probability of the situation that the instantaneous asymptotic BER exceeds a certain threshold, which is given by
\begin{equation}\label{eq:poutdefine}
P_{\mathrm{out}}=\Pr\left\{\tilde{P}_{b}^{\mathrm{CG-sub}}\geq\zeta\right\}.
\end{equation}
\end{definition}

\begin{theorem}
The outage probability can be computed in closed-form as
\begin{align}\label{eq:outage}
P_{\mathrm{out}}
&=\sum_{m=0}^\infty\frac{\rho^{2m}(1-\rho^2)}{m!}
\gamma\left(m+1,\frac{\lambda_1}{(1-\rho^2)\sigma_{h_0}^2}\right)+
\exp\left(\frac{-\lambda_2}{(1-\rho^2)\sigma_{h_1}^2}\right)\sum_{m=0}^\infty\sum_{n=0}^m\sum_{k=0}^n
\nonumber\\
&~~~\frac{\binom{n}{k}(-1)^k\rho^{2m}\lambda_1^{m+1}\lambda_2^n\sigma_{h_0}^{2k}\sigma_{h_1}^{2(m-n+k+1)}}
{m!n!(1-\rho^2)^{n-k-1}(\lambda_1\sigma_{h_1}^2-\lambda_2\sigma_{h_0}^2)^{m+k+1}}
\Gamma\left(m+k+1,\frac{\lambda_1\sigma_{h_1}^2-\lambda_2\sigma_{h_0}^2}{(1-\rho^2)\sigma_{h_0}^2\sigma_{h_1}^2}\right)-
\nonumber\\
&~~~\exp\left(\frac{-\lambda_1}{(1-\rho^2)\sigma_{h_1}^2}\right)
\sum_{m=0}^\infty\sum_{n=0}^m\sum_{k=0}^n
\frac{(m+k)!\binom{n}{k}(-1)^k\rho^{2m}\lambda_1^n\lambda_2^{m+1}\sigma_{h_0}^{2k}\sigma_{h_1}^{2(m-n+k+1)}}
{m!n!(1-\rho^2)^{n-k-1}\left(\lambda_2\sigma_{h_1}^2-\lambda_1\sigma_{h_0}^2\right)^{m+k+1}},
\end{align}
where
\begin{align}
\lambda_1=\frac{2Q^{-1}(\zeta)}{\gamma\left(\sqrt{N}-Q^{-1}(\zeta)\right)}~~~~~~ \textrm{and}~~~~~~\lambda_2=\frac{-2Q^{-1}(\zeta)}{\gamma\left(\sqrt{N}+Q^{-1}(\zeta)\right)}.
\end{align}
\end{theorem}

\begin{proof}
Substituting (\ref{eq:cgapb}) in (\ref{eq:poutdefine}), $P_{\mathrm{out}}$ is further given by
\begin{align}\label{eq:outage1}
P_{\mathrm{out}}
&=\Pr\left\{Q\left(\frac{\sqrt{N}\Delta}{\Sigma+2/\gamma}\right)\geq \zeta\right\}
=\Pr\left\{-\frac{Q^{-1}(\zeta)}{\sqrt{N}}\leq\frac{|h_0|^2-|h_1|^2}{|h_0|^2+|h_1|^2+2/\gamma}\leq \frac{Q^{-1}(\zeta)}{\sqrt{N}}\right\}
\nonumber\\
&=\Pr\left\{\left(1-\frac{Q^{-1}(\zeta)}{\sqrt{N}}\right)|h_0|^2-\left(1+\frac{Q^{-1}(\zeta)}{\sqrt{N}}\right)|h_1|^2\leq \frac{2Q^{-1}(\zeta)}{\gamma\sqrt{N}}\right.,
\nonumber\\
&~~~~~~~~~~~~~~~~~~~~~~~~\left.\left(1-\frac{Q^{-1}(\zeta)}{\sqrt{N}}\right)|h_1|^2-\left(1+\frac{Q^{-1}(\zeta)}{\sqrt{N}}\right)|h_0|^2\leq \frac{2Q^{-1}(\zeta)}{\gamma\sqrt{N}}\right\},
\end{align}
where $Q^{-1}(\cdot)$ denotes the inverse $Q$-function.
\begin{figure}
  \centering
  \includegraphics[width=80mm]{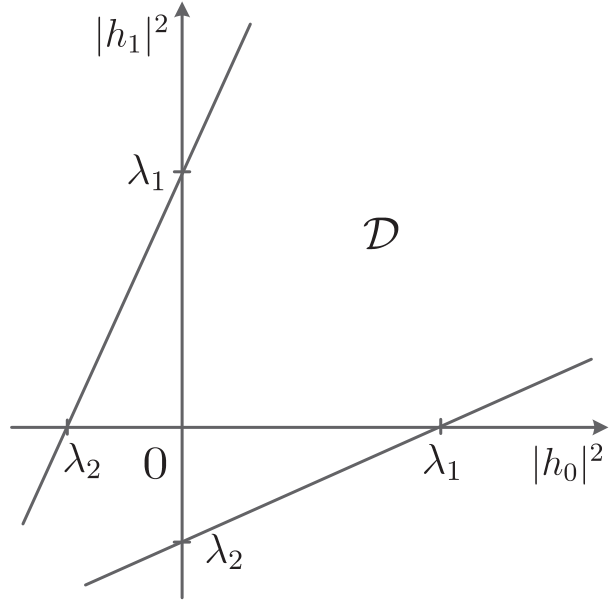}\\
  \caption{The domain of integration $\mathcal{D}$ for the calculation of the outage probability.}\label{fig:outageregion}
\end{figure}

Moreover, since $1-\frac{Q^{-1}(\zeta)}{\sqrt{N}}>0$, namely $\zeta>Q(\sqrt{N})$ generally holds for large $N$,  we have
\begin{align}\label{eq:pout}
P_{\mathrm{out}}
&=\iint_\mathcal{D} f_{|h_0|^2,|h_1|^2}(y_1,y_2)\mathrm{d}y_1\mathrm{d}y_2
\nonumber\\
&=\int_0^{\lambda_1}\int_0^{-\frac{\lambda_1 y_1}{\lambda_2}+\lambda_1}
f_{|h_0|^2,|h_1|^2}(y_1,y_2)\mathrm{d}y_2\mathrm{d}y_1
+\int_{\lambda_1}^\infty\int_{-\frac{\lambda_2 y_1}{\lambda_1}+\lambda_2}^{-\frac{\lambda_1 y_1}{\lambda_2}+\lambda_1}
f_{|h_0|^2,|h_1|^2}(y_1,y_2)\mathrm{d}y_2\mathrm{d}y_1
\nonumber\\
&\triangleq J_1(\zeta)+J_2(\zeta),
\end{align}
where the domain of integration $\mathcal{D}$ is displayed in Fig. \ref{fig:outageregion}, while $f_{|h_0|^2,|h_1|^2}(y_1,y_2)$ is the joint PDF of $|h_0|^2$ and $|h_1|^2$. The calculation of the integral $J_1(\zeta)$ and $J_2(\zeta)$ is given in Appendix \ref{sec:appendix0}.
\end{proof}

As channel affects BER performance, it is then of interest to check how the asymptotic  BER \eqref{eq:pbfloor} would satisfy a predefined performance under the random channel effect. We then define the asymptotic outage (AT) probability as the probability of the situation that the instantaneous BER floor falls below a certain threshold.

\begin{definition}
Define the asymptotic outage (AT) probability as
\begin{equation}\label{eq:patdefine}
P_{\mathrm{AT}}=\Pr\left\{P_{b}^{\mathrm{floor}}\geq\eta\right\}.
\end{equation}
\end{definition}

\begin{theorem}
The AT probability can be expressed as
\begin{align}\label{eq:atoutage}
\!\!P_{\mathrm{AT}}
=\!\sum_{m=0}^\infty\frac{C_{4m}x^{m+1}}{m+1}\!\left[{}_2F_1\!\left(\!2m\!+\!2,m\!+\!1;m\!+\!2,-\frac{\lambda_4}{\rho}\!\right)\!-
\!{}_2F_1\!\left(\!2m\!+\!2,m\!+\!1;m\!+\!2,-\frac{\lambda_3}{\rho}\!\right)\!\right],
\end{align}
where ${}_2F_1(\cdot,\cdot;\cdot,\cdot)$ denotes the Gauss hypergeometric function \cite{refs:toisap}, and
\begin{align}
\lambda_3=\frac{2}{1+\frac{Q^{-1}(\eta)}{\sqrt{N}}}-1,~~~~~~~~
\lambda_4=\frac{2}{1-\frac{Q^{-1}(\eta)}{\sqrt{N}}}-1.
\end{align}
\end{theorem}

\begin{proof}
Substituting (\ref{eq:pbfloor}) in (\ref{eq:patdefine}), $P_{\mathrm{AT}}$ is further given by
\begin{align}\label{eq:atoutage1}
P_{\mathrm{AT}}
=\Pr\left\{\frac{\Delta}{\Sigma}\leq \frac{Q^{-1}(\eta)}{\sqrt{N}}\right\}
=\Pr\left\{\frac{\left|\frac{|h_0|^2}{|h_1|^2}-1\right|}{\frac{|h_0|^2}{|h_1|^2}+1}\leq \frac{Q^{-1}(\eta)}{\sqrt{N}}\right\}.
\end{align}

Define $X=\frac{|h_0|^2}{|h_1|^2}$ whose cumulative density function (CDF) can be computed as
\begin{align}\label{eq:hypergeometric}
&F_X(x)=\sum_{m=0}^\infty\frac{C_{4m}x^{m+1}}{m+1}{}_2F_1\left(2m+2,m+1;m+2,-\frac{x}{\rho}\right),
\end{align}
where the detailed derivation can be found in Appendix \ref{sec:appendix1}. Then we have
\begin{align}\label{eq:pat}
P_{\mathrm{AT}}=\Pr\left\{-\frac{Q^{-1}(\eta)}{\sqrt{N}}\leq \frac{X-1}{X+1}\leq \frac{Q^{-1}(\eta)}{\sqrt{N}}\right\}
=\Pr\left\{\lambda_3\leq X\leq \lambda_4\right\}
=F_X(\lambda_4)-F_X(\lambda_3).
\end{align}
Thus, the AT probability is obtained by substituting the CDF of $X$ into (\ref{eq:pat}).
\end{proof}

\subsection{Suboptimal Detector with the PSK Ambient Source}
In practice, ambient RF signals are usually the PSK or the Quadrature Amplitude Modulation (QAM) signals rather than the complex Gaussian signal. In this section,  we will study the suboptimal detector and its performance with PSK ambient signals\footnote{The extension to QAM ambient signal can be similarly made and is omitted due to the length limit.}, i.e.,
\begin{align}
s[n]=\sqrt{P_s}\exp\left(\mathrm{j}\frac{2\pi k}{M}\right),~~~~~~~~k=0,\cdots,M-1,
\end{align}
where  $P_s$ is the signal power.

Let us  explicitly expand $Z$ as
\begin{align}
Z =\left\{
   \begin{array}{ll}
   \sum\limits_{n=1}^{N}\left(|h_0|^2|s[n]|^2 +|w[n]|^2+ 2\Re\{h_0s[n]w^*[n] \}\right),&\mathcal{H}_0,\\[2 mm ]
   \sum\limits_{n=1}^{N}\left(|h_1|^2|s[n]|^2 +|w[n]|^2+ 2\Re\{h_1s[n]w^*[n] \}\right),&\mathcal{H}_1,
   \end{array}\right.
\end{align}

From the central limit theorem, we have $|w[n]|^2\sim\mathcal{N}(N_w,N_w^2)$ and $\Re\{h_is[n]w^*[n] \}\sim\mathcal{N}(0,|h_i|^2P_sN_w)$. Then the distribution of $Z$ under the hypothesis $\mathcal{H}_i$ can be obtained as $Z|\mathcal{H}_i\sim\mathcal{N}(\mu^{\mathrm{PSK}}_i,\varsigma^{\mathrm{PSK}}_i)$, with the PDF
\begin{align}
\hat{f}_{Z}(z|\mathcal{H}_i)=\frac{1}{\sqrt{2\pi\varsigma_i^{\mathrm{CG}}}}\exp\left[-\frac{\left(z-\mu^{\mathrm{CG}}_i\right)^2}{2\varsigma_i^{\mathrm{CG}}}\right],~~~~~~~~i=0,1,
\end{align}
where
\begin{align}\label{eq:hatmupsk}
\mu^{\mathrm{PSK}}_i=N\sigma_i^2,~~~~~~\varsigma^{\mathrm{PSK}}_i=2N|h_i|^2P_s N_w+NN_w^2,~~~~~~i=0,1.
\end{align}

\begin{theorem}
The threshold for the suboptimal detector with PSK ambient signals is expressed as
\begin{align}\label{eq:thmpsk}
T_h^{\mathrm{PSK}}&=\frac{NN_w}{2}+NN_w\sqrt{\left(|h_0|^2\gamma+\frac{1}{2}\right)\left(|h_1|^2\gamma+\frac{1}{2}\right)
\left[1+\frac{2\ln\left(\frac{2|h_0|^2\gamma+1}{2|h_1|^2\gamma+1}\right)}{N\gamma(|h_0|^2-|h_1|^2)}\right]}.
\end{align}
\end{theorem}

\begin{proof}
Similar to the operation (\ref{eq:mldetector}), the optimum threshold for locating the range of the energy $Z$ is obtained through $\hat{f}_{Z}\left(T_h^{\mathrm{PSK}}|\mathcal{H}_0\right)=\hat{f}_{Z}\left(T_h^{\mathrm{PSK}}|\mathcal{H}_1\right)$. After some tedious yet straightforward  calculation, we will obtain the result in (\ref{eq:thmpsk}).
\end{proof}

\begin{theorem}
The BER for the suboptimal detector can be expressed as
\begin{align}\label{eq:pskpb}
P_{b}^{\mathrm{PSK}}=\frac{1}{2}-\frac{1}{2}Q\left(\frac{T_{h}^{\mathrm{PSK}}-\mu_{\max}}{\sqrt{\varsigma_{\max}}}\right)+
\frac{1}{2}Q\left(\frac{T_{h}^{\mathrm{PSK}}-\mu_{\min}}{\sqrt{\varsigma_{\min}}}\right),
\end{align}
where $\mu_{\max}=\max\left\{\mu^{\mathrm{PSK}}_0,\mu^{\mathrm{PSK}}_1\right\}$,
$\mu_{\min}=\min\left\{\mu^{\mathrm{PSK}}_0,\mu^{\mathrm{PSK}}_1\right\}$,
$\varsigma_{\max}=\max\left\{\varsigma^{\mathrm{PSK}}_0,\varsigma^{\mathrm{PSK}}_1\right\}$ and $\varsigma_{\min}=\min\left\{\varsigma^{\mathrm{PSK}}_0,\varsigma^{\mathrm{PSK}}_1\right\}$.
\end{theorem}

\begin{proof}
The proof is similar to that of Theorem \ref{theorem:cgsubber}.
\end{proof}

We can see that the threshold (\ref{eq:thmpsk}) cannot be obtained without the knowledge of CSI. However, if the reader have access to the knowledge of the noise, i.e., $N_w$, we can obtain the threshold with $\sigma_i^2$ as follows
\begin{align}\label{eq:pskth}
T_h^{\mathrm{PSK}}&=\frac{NN_w}{2}+\frac{N}{2}\sqrt{\left(2\sigma_0^2-N_w\right)\left(2\sigma_1^2-N_w\right)
\left[1+\frac{2N_w\ln\left(\frac{2\sigma_0^2-N_w}{2\sigma_1^2-N_w}\right)}{N(\sigma_0^2-\sigma_1^2)}\right]}.
\end{align}

Nevertheless, we provide another solution even when $N_w$ is unknown.
\begin{corollary}
For high SNR  circumstance with $2|h_i|^2P_s+N_w\gg N_w$ and with large $N$, the asymptotic threshold is expressed as
\begin{align}\label{eq:psktha}
\tilde{T}_h^{\mathrm{PSK}}=  N\sigma_0\sigma_1.
\end{align}
\end{corollary}
\begin{proof}
When there is $2|h_i|^2P_s+N_w\gg N_w$, the asymptotic distribution of $Z$ with PSK ambient signals under the hypothesis $\mathcal{H}_i$ can be approximated by
\begin{align}\label{eq:pskzdistri}
Z|\mathcal{H}_i \sim\mathcal{N}\left(N|h_i|^2P_s+NN_w,2N|h_i|^2P_s N_w+2NN_w^2\right)=\mathcal{N}\left(N\sigma_i^2,2NN_w\sigma_i^2\right).
\end{align}
Similar to the operation before, the corresponding threshold is  given by
\begin{align}\label{eq:pskth2}
\tilde{T}_h^{\mathrm{PSK}}= N\sigma_0\sigma_1\sqrt{1+\frac{2N_w\ln\left(\frac{\sigma_0^2}{\sigma_1^2}\right)}{N(\sigma_0^2-\sigma_1^2)}}\approx N\sigma_0\sigma_1,
\end{align}
where the approximation holds valid for $N$ large enough. Then the threshold $\tilde{T}_h^{\mathrm{PSK}}$ can be obtained just with knowledge of  $\sigma_i^2$.
\end{proof}

\begin{rem}
The proposed suboptimal detector with PSK ambient signals achieves the balanced BER for $d=0$ and $d=1$ at the threshold (\ref{eq:pskth2}), i.e.,
\begin{align}
&\Pr(\hat{d}=0|\mathcal{H}_1)-\Pr(\hat{d}=1|\mathcal{H}_0)
=\frac{1}{2}-\frac{1}{2}Q\left(\frac{\tilde{T}_{h}^{\mathrm{PSK}}-N\sigma^2_0}{\sqrt{2NN_w\sigma_0^2}}\right)-\frac{1}{2}Q\left(\frac{\tilde{T}_{h}^{\mathrm{PSK}}-N\sigma^2_1}{\sqrt{2NN_w\sigma_1^2}}\right)\nonumber\\
&=\frac{1}{2}\left[1-Q\left(\frac{\sqrt{N}(\sigma_1-\sigma_0)}{\sqrt{2N_w}}\right)-Q\left(\frac{\sqrt{N}(\sigma_0-\sigma_1)}{\sqrt{2N_w}}\right)\right]
=0.
\end{align}
\end{rem}

\begin{corollary}
For high SNR  circumstance with $2|h_i|^2P_s+N_w\gg N_w$ and large $N$, the asymptotic BER is given by
\begin{align}\label{eq:pskpba}
\tilde{P}_b^{\mathrm{PSK}}=
Q\left(\sqrt{\frac{N}{2}}\left|\sqrt{|h_0|^2\gamma+1}-\sqrt{|h_1|^2\gamma+1}\right|\right)
\approx Q\left(\sqrt{\frac{N\gamma}{2}}\big||h_0|-|h_1|\big|\right).
\end{align}
\end{corollary}
\begin{proof}
The result is easily obtained by recomputing (\ref{eq:pskpb}), i.e., replacing $T_h^{\mathrm{PSK}}$ with $\tilde{T}_h^{\mathrm{PSK}}$,
and replacing $\sigma_i^{\rm{PSK}}$ with $2NN_w\sigma_i^2$,
\begin{align}\label{eq:pskpba2}
\!\!\!\!\tilde{P}_b^{\mathrm{PSK}}\!
=\frac{1}{2}-\frac{1}{2}Q\!\left(\frac{\tilde{T}_h^{\mathrm{PSK}}-N\sigma_{\max}^2}{\sqrt{2NN_w\sigma_{\max}^2}}\right)
+\frac{1}{2}Q\!\left(\frac{\tilde{T}_h^{\mathrm{PSK}}-N\sigma_{\min}^2}{\sqrt{2NN_w\sigma_{\min}^2}}\right)
=Q\!\left(\frac{\sqrt{N}|\sigma_0-\sigma_1|}{\sqrt{2N_w}}\right).
\end{align}
\end{proof}

Unlike the case of complex Gaussian ambient signals, the BER (\ref{eq:pskpba}) with PSK ambient signals is not only an decreasing function of SNR  but also meets no error floor as SNR goes to infinity. It is also noted that the channel difference $\big||h_0|-|h_1|\big|$ rather than RCD  affects the performance here. Moreover, increasing the sampling number $N$ has the same effect as increasing SNR.

\section{Parameter Estimation}
\label{sec:estimation}
For the proposed detectors (\ref{eq:optimalth}) and (\ref{eq:suboptimalth}),  the reader does not need to estimate the channel state information of $h_{st}$, $h_{sr}$, and $h_{tr}$, as well as  $s[n]$ and $\alpha$. Nevertheless, the two crucial parameters $\sigma_0^2$ and $\sigma_1^2$ should be estimated before the detection.

\subsection{Blind Estimation of $\sigma_0^2$ and $\sigma_1^2$}
Since the channel energy (or equivalently the channel amplitude) varies much slower than the instantaneous CSI, we assume that the coherent time of  channel energy spans much longer than the channel coherent time. Specifically, let us assume the channel energy does not change during $M$ symbol periods of the tag, (or $MN$ $s[n]$'s correspondingly), and the corresponding received signal vectors at the reader are denoted as $\boldsymbol{y}_k~(k=1,\cdots,M)$. Bearing in mind that $\sigma_0^2$ and $\sigma_1^2$ represent the statistic variances of the received signal in (\ref{eq:y}), we then propose the following estimation steps:
\begin{enumerate}[Step 1:]
\item Compute the normalized energy of $\boldsymbol{y}_k$ as
\begin{align}
A_k=\frac{\|\boldsymbol{y}_k\|^2}{N},~~~~~~k=1,\cdots,M.
\end{align}
\item Arrange $A_k$ in ascending order, denoted as $A_k^{\uparrow},~k=1,\cdots,M$.
\item Since the tag transmits symbols of 0 and 1 with equal probability, average the first and second half of $A_k^{\uparrow}$ as
\begin{align}\label{eq:A}
\!\!A_{\min}=\frac{2}{M}\sum_{k=1}^{M/2}A_k^{\uparrow},~~~~~~~~A_{\max}=\frac{2}{M}\!\sum_{k=M/2+1}^{M}\!A_k^{\uparrow}.
\end{align}
\end{enumerate}

However, (\ref{eq:A}) can not tell which one of $A_{\min}$ and $A_{\max}$ corresponds to which $\sigma_i^2$.

\subsection{Discrimination of $\sigma_0^2$ and $\sigma_1^2$ with Short Training}
\begin{figure}[t]
\centering
\includegraphics[width=110mm]{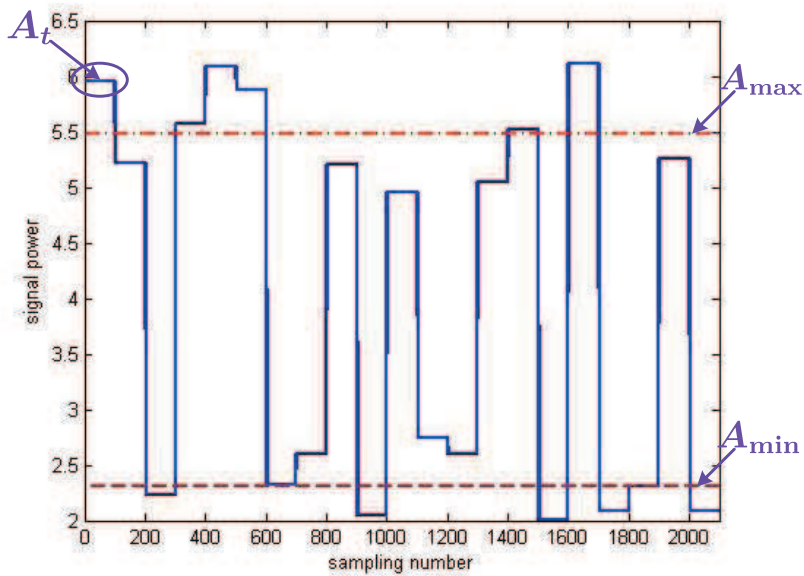}
\caption{An example demonstrating the estimation of $\sigma_0^2$ and $\sigma_1^2$, with $N=100$, $M=20$, and $M_t=1$.}
\label{fig:parameter estimation}
\end{figure}

We employ a very short training to discriminate $\sigma_0^2$ from $\sigma_1^2$. Assume the tag sends $M_t\geq 1$  bits as training symbols and the corresponding received signal vectors are $\boldsymbol{y}_{ti}~(i=1,\cdots,M_t)$. Then we continue the previous estimation approach as
\begin{enumerate}[Step 4:]
\item Compute the average of $M_t$ normalized powers as
\begin{align}
A_t=\frac{1}{M_t}\sum_{i=1}^{M_t}\frac{\|\boldsymbol{y}_{ti}\|^2}{N}.
\end{align}
\end{enumerate}
\begin{enumerate}[Step 5:]
\item If $|A_{\min}-A_t|<|A_{\max}-A_t|$, set $\hat{\sigma}_0^2=A_{\max}$ and $\hat{\sigma}_1^2=A_{\min}$; otherwise set $\hat{\sigma}_0^2=A_{\min}$ and $\hat{\sigma}_1^2=A_{\max}$.
\end{enumerate}

A specific example is presented here with $N=100$, $M=20$, and $M_t=1$. We show $A_t$ and $A_k~(k=1,\cdots,20)$ in Fig. \ref{fig:parameter estimation} and obtain $A_{\min}$ and $A_{\max}$ as the corresponding values of the two dotted lines. Since $|A_{\min}-A_t|>|A_{\max}-A_t|$, we set $\hat{\sigma}_0^2=A_{\min}$ and $\hat{\sigma}_1^2=A_{\max}$.

\begin{rem}
Theoretically, sending one training symbol is sufficient to distinguish $\sigma_0^2$ and $\sigma_1^2$. Moreover, we call this estimation a ``semi-blind'' method, where the energies of $M$ symbols are utilized to blindly estimate values of $\sigma_i^2$ while only few training symbols are required to differentiate between the two $\sigma_i^2$'s.
\end{rem}

\section{Numerical Results }
\label{sec:numerical results}
In this section, we resort to numerical examples to evaluate the proposed studies. Since the distance between the source and the tag (or the distance between the source and the reader) is much larger than that between the tag and the reader \cite{refs:abwcoota}, we generate the channels $h_{st}$ and $h_{sr}$ according to $\mathcal{CN}(0,1)$ and the channel $h_{tr}$ according to $\mathcal{CN}(0,10)$. Energies of all channels are assumed to hold unchanged during 50 symbol period  of the tag, i.e., $M=50$, and 4 training symbols of bit "1" are periodically inserted, i.e., $M_t=4$. The tag coefficient $\alpha=0.5$ and the AGWN follows $\mathcal{CN}(0,1)$. Totally $10^6$ Monte-Carlo runs are adopted for average.

\begin{figure}[t]
\centering
\includegraphics[width=110mm]{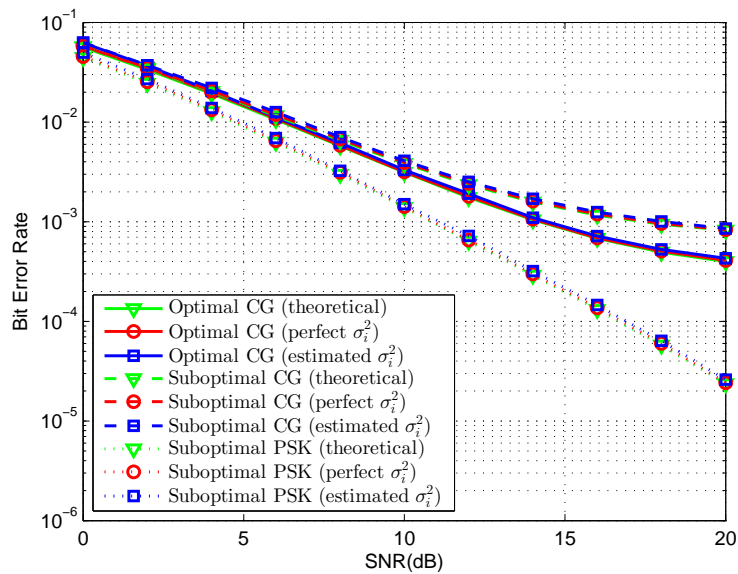}
\caption{BER versus SNR for the detectors with $N=40$ and RCD = 0.5. }
\label{fig:BERvsSNR}
\end{figure}
We first demonstrate the BER versus SNR of the proposed detectors in Fig. \ref{fig:BERvsSNR}. The simulated BERs with perfect $\sigma_i^2$ and estimated $\sigma_i^2$ are displayed, respectively, and thresholds of different detectors in (\ref{eq:optimalth}), (\ref{eq:suboptimalth}) and (\ref{eq:pskth}) are all applied for simulation. The theoretical results in (\ref{eq:optimalpb}), (\ref{eq:gaussianpb}) and (\ref{eq:pskpb}) are also shown for comparison. We set $N=50$ and RCD = 0.5. It is seen that for all cases, the simulated BERs with perfect $\sigma_i^2$ are consistent with the theoretical BER. Moreover, the simulated BER with estimated $\sigma_i^2$ performs ignorably worse than that with perfect $\sigma_i^2$, which indicates the effectiveness of the proposed estimation approach in Section \ref{sec:estimation}. For the complex Gaussian (CG) ambient signal, the optimal detector outperforms the suboptimal one, as expected, and higher SNR leads to smaller BER while the performance improvement will flatten as SNR becomes relatively large, which verifies (\ref{eq:cgapb}). However, for the PSK ambient signal, it achieves better performance than the CG, since $\sqrt{\gamma}$ is in the numerator of (\ref{eq:pskpb}), while the effect of $\gamma$ on BER is partly alleviated by $\Sigma$ as shown in (\ref{eq:cgapb}). Moreover, there exists no error floor as SNR becomes larger, as analyzed in (\ref{eq:pskpba}).

\begin{figure}[t]
\centering
\includegraphics[width=110mm]{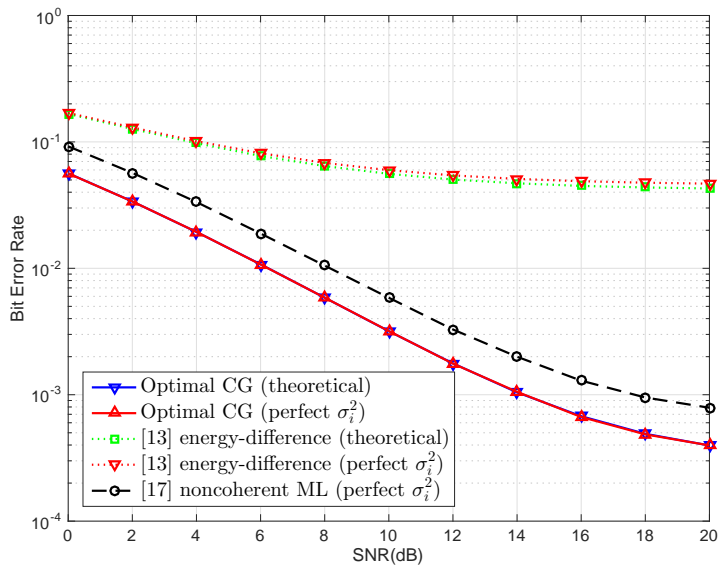}
\caption{Performance comparison between semi-coherent and non-coherent detectors with $N=40$, RCD = 0.5, and the CG ambient source. }
\label{fig:semivsnon}
\end{figure}
We then compare the performance of the semi-coherent detector with that of the existing noncoherent detectors in Fig. \ref{fig:semivsnon}, where $N=40$, RCD = 0.5 and the ambient source transmits CG signals. Specifically, the theoretical and simulated BERs of our optimal detector and the energy-difference method in \cite{refs:udabafabcs}, and the simulated BER of the noncoherent ML detector in \cite{refs:sdaoabswdm} are demonstrated, respectively, for comparison. All the simulated BERs are obtained with perfect $\sigma_i^2$. We can see that the optimal semi-coherent detector outperforms the noncoherent ones, at all SNR region.

\begin{figure}[t]
\centering
\includegraphics[width=110mm]{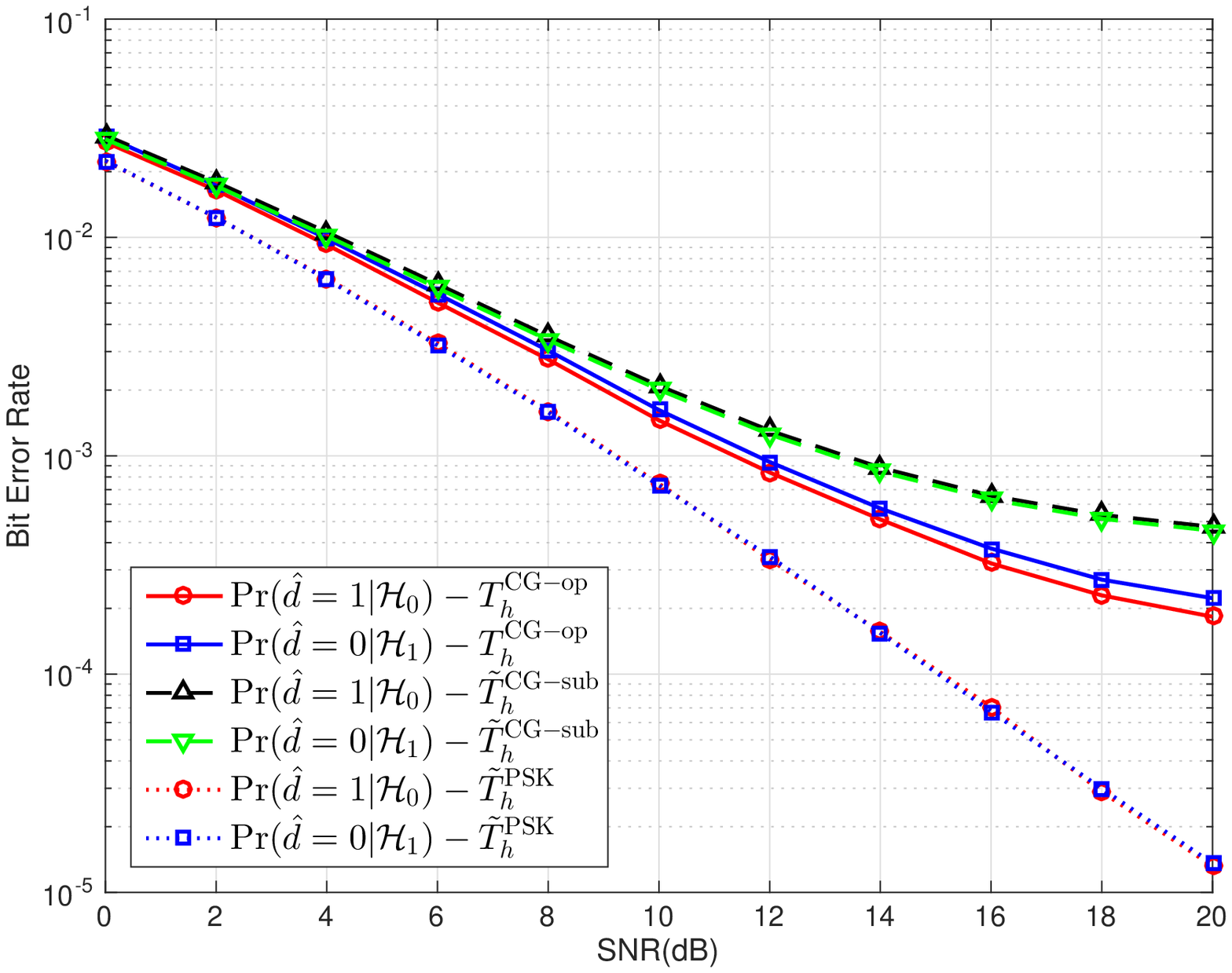}
\caption{Balanced or unbalanced phenomenon for the detectors with $N=40$ and RCD = 0.5. }
\label{fig:balanced BER}
\end{figure}
The balanced or unbalanced BER phenomenon of the proposed detectors is then illustrated in Fig. \ref{fig:balanced BER}, where we set $N=40$ and RCD = 0.5, and $\Pr(\hat{d}=1|\mathcal{H}_0)$ and $\Pr(\hat{d}=0|\mathcal{H}_1)$ corresponding to the thresholds (\ref{eq:optimalth}), (\ref{eq:ta}) and (\ref{eq:pskth2}) are simulated. In order to more clearly illustrate the phenomena, all the thresholds are only computed with perfect $\sigma_i^2$. As analyzed previously, both (\ref{eq:ta}) and (\ref{eq:pskth2}) can achieve the balanced BER for ``0'' and ``1'' while (\ref{eq:optimalth}) can not.

\begin{figure}[t]
\centering
\includegraphics[width=110mm]{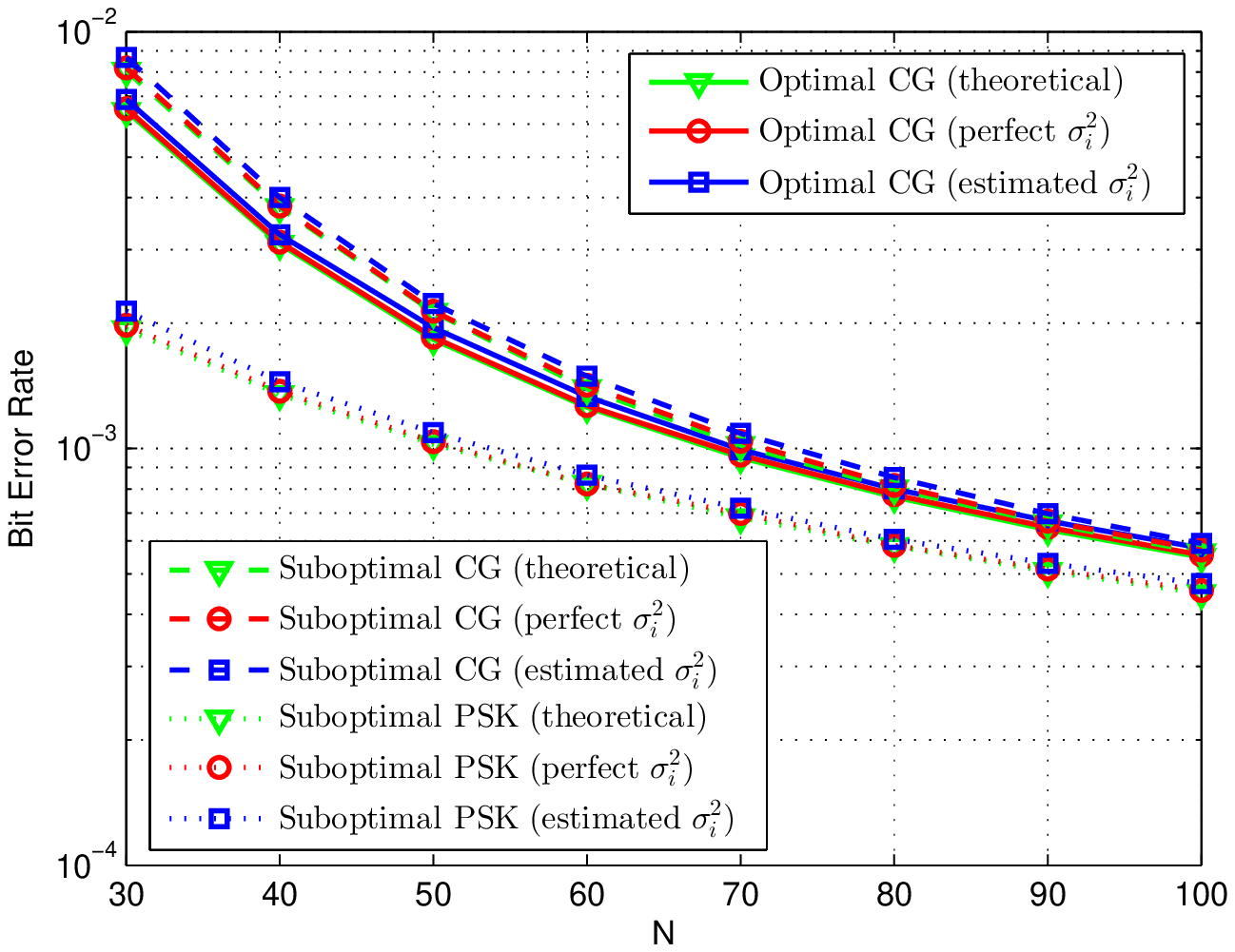}
\caption{BER versus $N$ for the detectors with SNR = 10 dB and RCD = 0.5.}
\label{fig:BERvsN}
\end{figure}
We next show the BER versus the length of the received signal vector, $N$, for the detectors in Fig. \ref{fig:BERvsN}. We set SNR = 10 dB and RCD = 0.5.  Similar to Fig. \ref{fig:BERvsSNR}, the curves of the theoretical BER, simulated BER with perfect $\sigma_i^2$ and simulated BER with estimated $\sigma_i^2$ are all close to each other. It is obviously seen that larger $N$ results in a reduced BER for all the detectors and there is no error floorwhen $N$ increases as seen from the theoretical expression (\ref{eq:cgapb}) and (\ref{eq:pskpba}). Nevertheless, in practice one cannot use very large $N$ since it will decrease the transmission rate of tag's symbols, increase the computational complexity, and may exceed the channel energy coherence time. In addition, the suboptimal detector with CG ambient signals performs closer to the optimal one since the Gaussian approximation utilized in the suboptimal detector works better at larger $N$. Moreover, the detector with CG ambient signals performs closer to that with PSK ambient signals as $N$ becomes large, because the distribution of $Z$ with CG ambient signals approximates to that with PSK ambient signals, both locating around $N\sigma_i^2$ with a relatively large probability as shown in (\ref{eq:hatmu}) and (\ref{eq:hatmupsk}).

\begin{figure}[t]
\centering
\includegraphics[width=110mm]{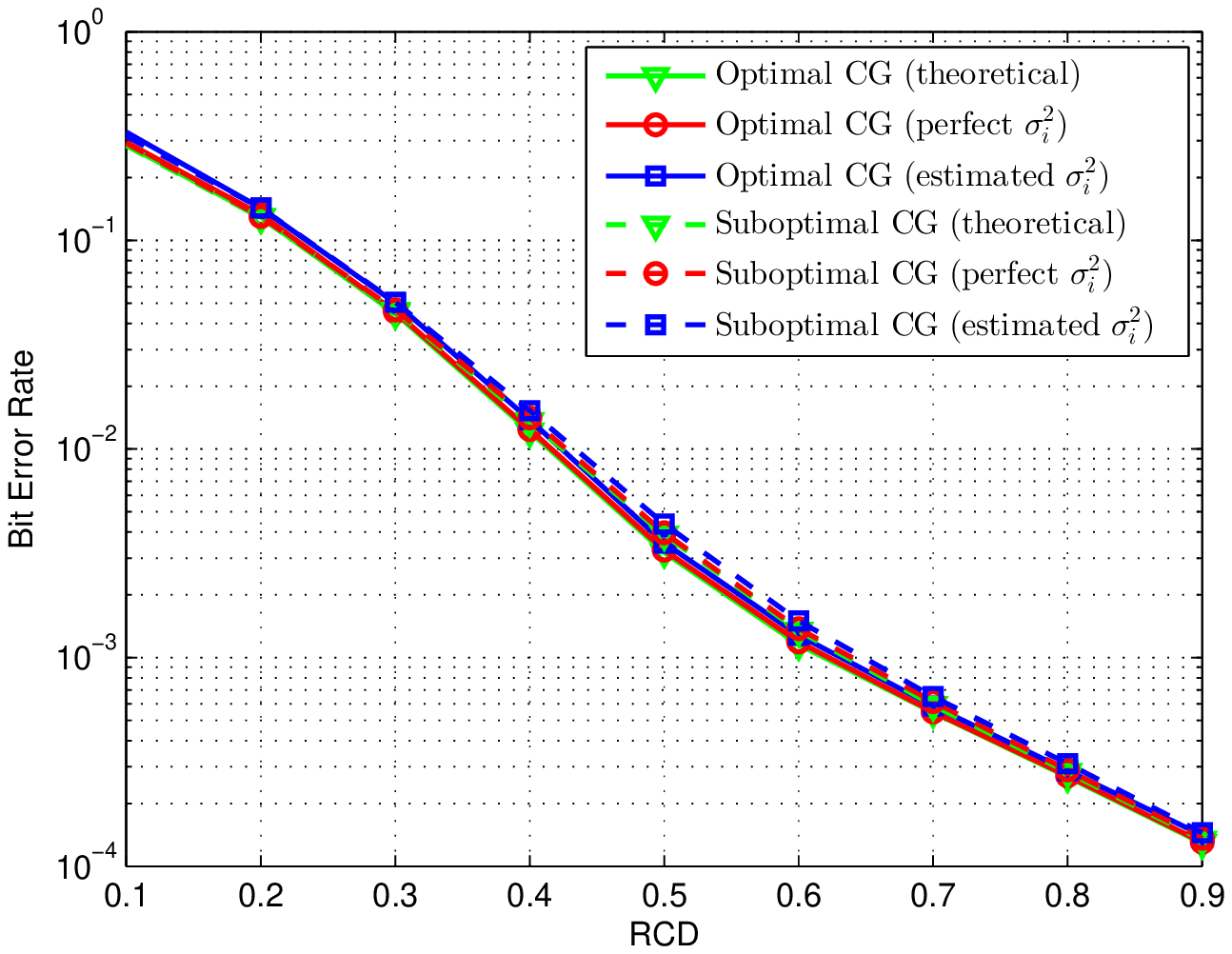}
\caption{BER versus RCD for the detectors with SNR = 10 dB and $N=40$.}
\label{fig:BERvsRCD}
\end{figure}
Fig. \ref{fig:BERvsRCD} depicts the curves of BER versus RCD corresponding to  the optimal and suboptimal detectors with CG ambient signals. We set SNR = 10 dB and $N=40$. Obviously, large RCD  results in smaller BER and there is no error floor effect, which is intuitively correct since the reader can easily decode the symbol when the channels corresponding to ``0'' and ``1'' are relatively distinct. Compared with the BER values in Fig. \ref{fig:BERvsSNR} and Fig. \ref{fig:BERvsN}, we can infer that RCD has a more important impact on BER performance than other system parameters. The improvement of the performance is gradual at small RCD but becomes rapid at large RCD, because the effect of SNR may slow down the reduction of BER at small RCD, while larger RCD will totally dominate the BER, as can be verified from (\ref{eq:cgapb}). It can also be seen that the BERs approach to 0.5 at small RCD, since both the detectors fail to work with the poorest detection environment and only yield random results.

\begin{figure}[t]
\centering
\includegraphics[width=110mm]{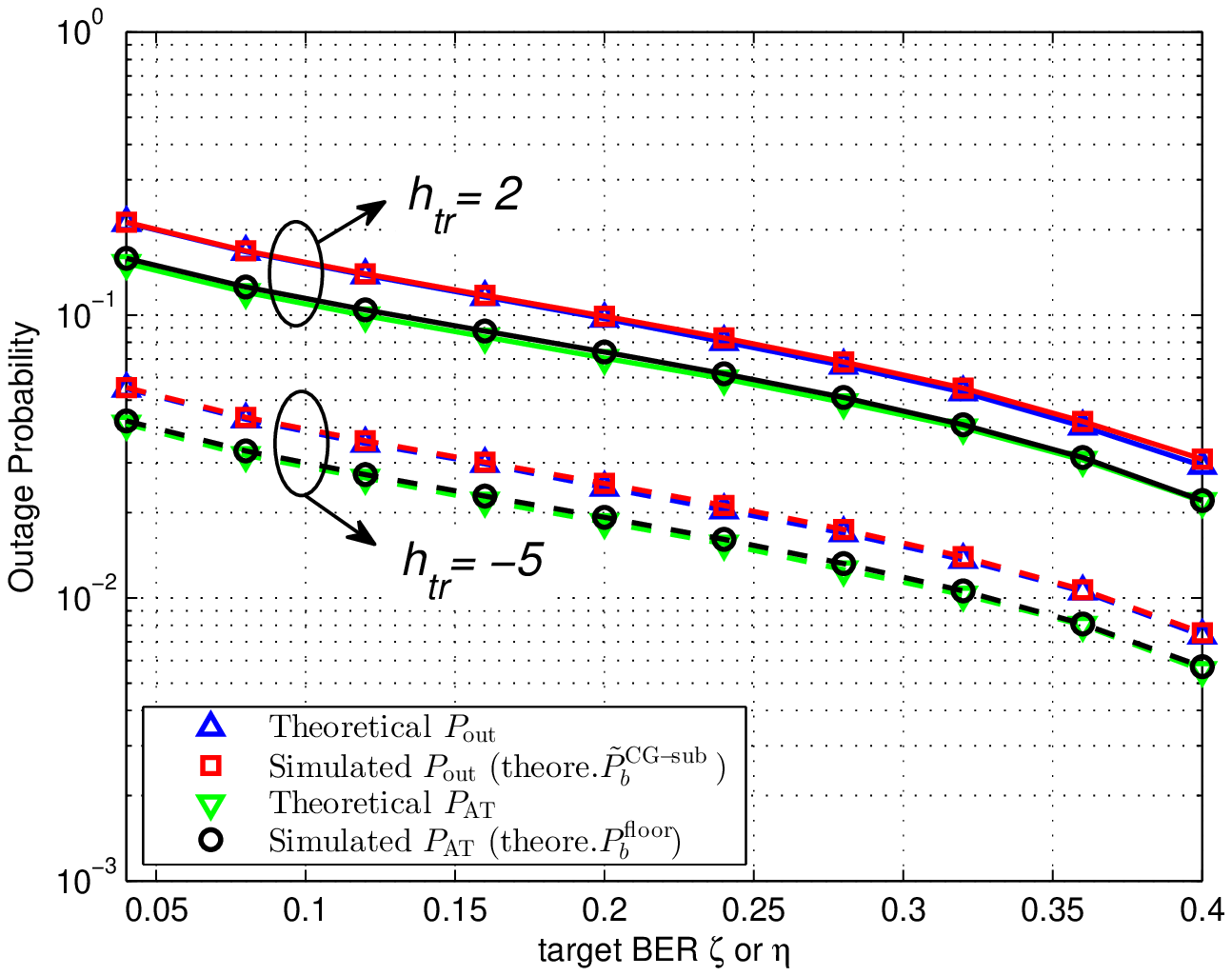}
\caption{Outage probability and AT probability versus target BER for the suboptimal detector with SNR = 20 dB and $N=40$.}
\label{fig:outage}
\end{figure}
\begin{figure}[t]
\centering
\includegraphics[width=110mm]{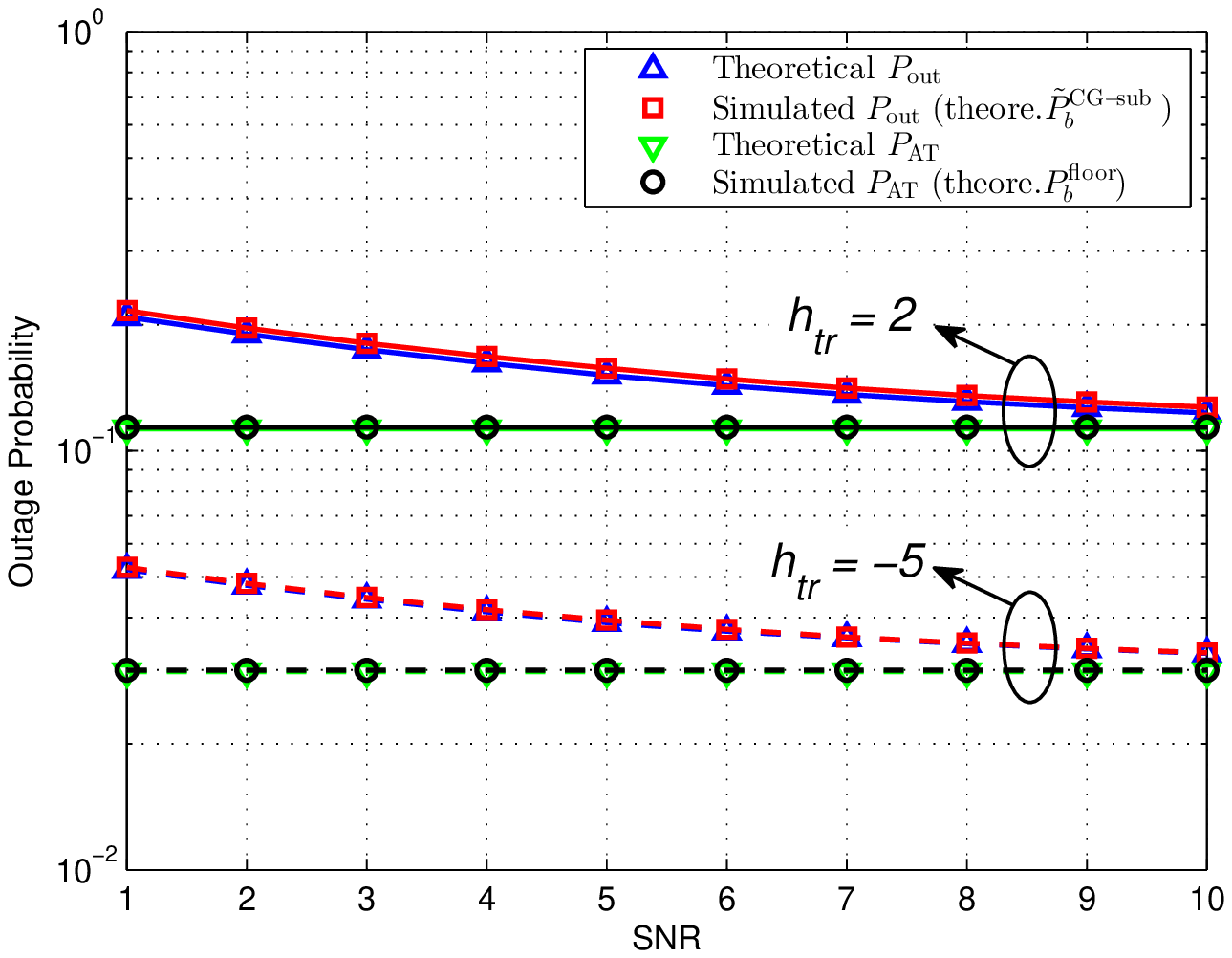}
\caption{Outage probability and AT probability versus SNR for the suboptimal detector with $N=40$ and $\zeta=\eta=0.1$, $h_{tr}$ is set as 2 and -5 for comparison.}
\label{fig:outagevssnr}
\end{figure}
In this example, we illustrate the outage probability and AT probability of the suboptimal detector versus the target BER in Fig. \ref{fig:outage} and those versus SNR in Fig. \ref{fig:outagevssnr}. In Fig. \ref{fig:outage}, the parameters are set as SNR = 5 dB and $N = 40$, while in Fig. \ref{fig:outagevssnr}, we set $N = 40$ and $\zeta=\eta=0.1$. Since $h_{tr}$ is assumed as a constant during the outage derivation in Appendix \ref{sec:appendix0}, we set $h_{tr} = 2$ and $h_{tr} = -5$ for comparison. The theoretical BERs in (\ref{eq:cgapb}) and (\ref{eq:pbfloor}) are employed for outage simulation. The theoretical outage probability given by (\ref{eq:outage}) and (\ref{eq:atoutage}) is displayed as well. As can be seen, the theoretical analysis matches the results of the Monte Carlo runs very well.
Naturally, a larger target BER leads to a lower outage probability. As mentioned in Fig. \ref{fig:BERvsSNR}, BER approaches an SNR-independent error floor as SNR turns large, while the outage probability correspondingly flattens and approaches the AT probability. Meanwhile, $h_{tr}$ with larger absolute value can achieve lower AT probability since larger $|h_{tr}|$ will amplify the difference between $|h_0|$ and $|h_1|$, i.e., the RCD or the correlation coefficient $\rho$ in (\ref{eq:rho}), which would contribute to a better outage performance.

\begin{figure}[t]
\centering
\includegraphics[width=110mm]{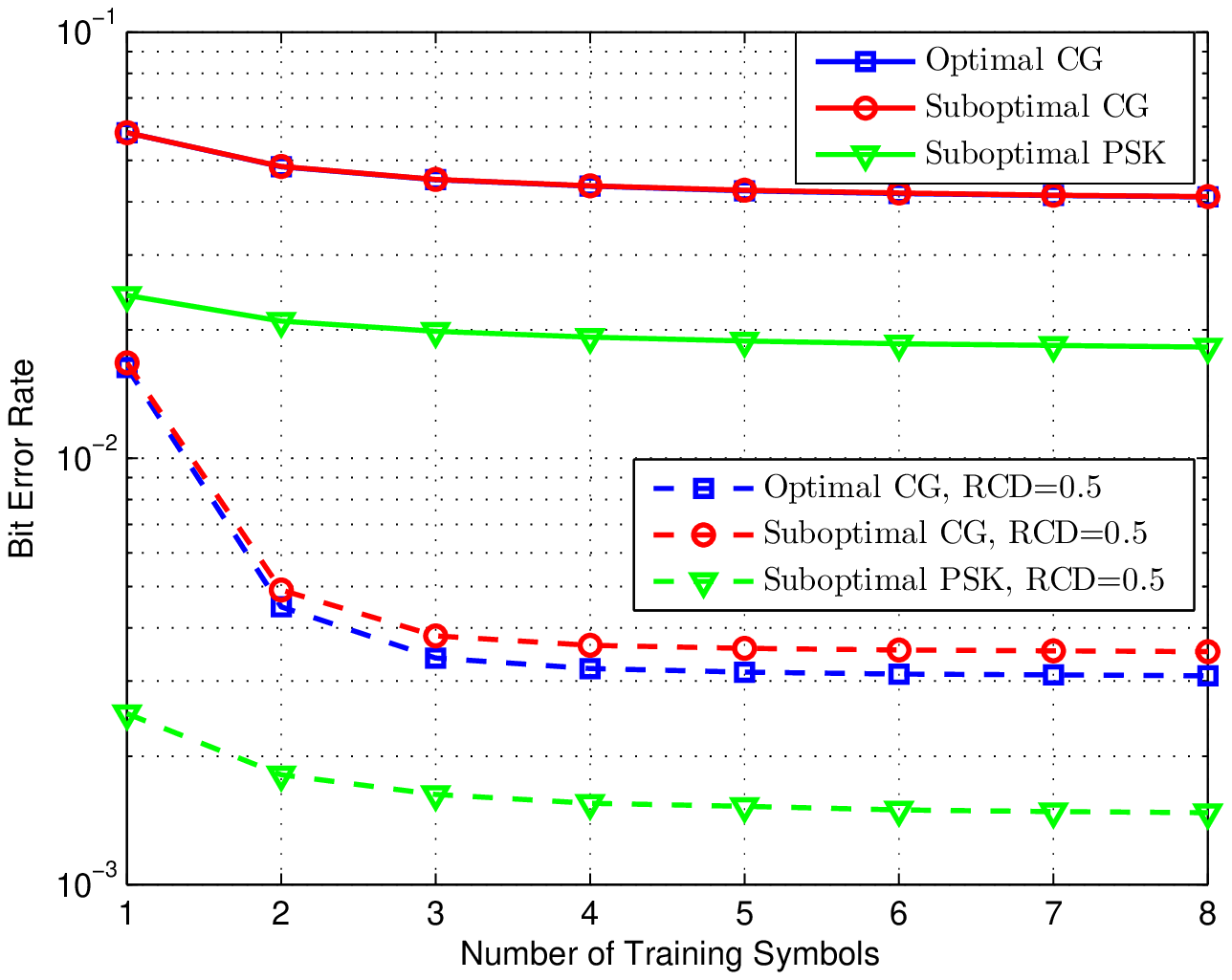}
\caption{BER versus the number of training symbols for the three detectors with SNR = 10 dB and $N=40$. RCD is unconstrained and set as 0.5, respectively.}
\label{fig:BERvsTrainNum}
\end{figure}
Lastly, we demonstrate simulated BER versus the number of training symbols in Fig. \ref{fig:BERvsTrainNum} when three detection thresholds (\ref{eq:optimalth}), (\ref{eq:suboptimalth}) and (\ref{eq:pskth}) are applied for comparison. We set SNR = 10 dB, $N = 40$.  The RCD is unconstrained and set as 0.5 for comparison. We can see that, on one hand, sending more training symbols contributes to a better BER performance, especially when the number turns from 1 to 2; on the other hand, no more distinct performance improvement can be achieved by keeping increasing the number of training symbols. Hence, 3 or 4 training symbols are appropriate for the comprehensive consideration of system performance and complexity.

\section{Conclusion }
\label{sec:conclusion}
This paper presents a theoretical study of the semi-coherent detection for the ambient backscatter system, where training symbols are sent to acquire the detection-required parameters rather than the channels themselves. Our goal is to offer feasible suggestions for practical system designs of this new born communication prototype. We proposed designed symbol detectors  under different scenarios to realize the trade-off between the detection accuracy and the freedom from prior knowledge. The closed-form BER expressions and outage analysis are also derived for various cases, which demonstrate the effect of different system parameters. Simulation results are provided to verify the correctness of our studies.

\appendices
\section{Calculation of the integral $J_1(\zeta)$ and $J_2(\zeta)$}\label{sec:appendix0}
Consider the situation where the distance between the tag and the reader is much smaller than that between the tag and the source (or the reader and the source), and the communication environment around the tag and the reader is usually stationary during the data transmission, the channel coefficient $h_{tr}$ can be taken as a constant. Then we regard $h_1=h_0+\alpha h_{st}h_{tr}$ as the sum of two independent zero-mean complex Gaussian random variables, i.e., a new zero-mean complex Gaussian random variable whose variance is $\sigma_{h_1}^2=\sigma_{h_0}^2+\alpha^2 |h_{tr}|^2\sigma_{st}^2$. Since $|h_0|^2$ and $|h_1|^2$ are correlated, their joint PDF is given by
\begin{align}\label{h0h1jointpdf}
f_{|h_0|^2,|h_1|^2}(y_1,y_2)
=\frac{1}{(1-\rho^2)\sigma_{h_0}^2\sigma_{h_1}^2}\exp\left[-\frac{1}{1-\rho^2}\left(\frac{y_1}{\sigma_{h_0}^2}+\frac{y_2}{\sigma_{h_1}^2}\right)\right]
I_0\left(\frac{2\rho\sqrt{y_1y_2}}{\sigma_{h_0}\sigma_{h_1}(1-\rho^2)}\right),
\end{align}
where $\rho$ is the correlation coefficient between $|h_0|^2$ and $|h_1|^2$ with the form
\begin{align}\label{eq:rho}
\rho=\frac{\mathrm{E}\{|h_0|^2|h_1|^2\}-\mathrm{E}\{|h_0|^2\}\mathrm{E}\{|h_1|^2\}}{\sqrt{\mathrm{D}\{|h_0|^2\}}\sqrt{\mathrm{D}\{|h_1|^2\}}}
=\frac{(2\sigma_{h_0}^4+\sigma_{h_0}^2\sigma_f^2)-\sigma_{h_0}^2(\sigma_{h_0}^2+\sigma_f^2)}{\sigma_{h_0}^2(\sigma_{h_0}^2+\sigma_f^2)}
=\frac{\sigma_{h_0}^2}{\sigma_{h_1}^2},
\end{align}
and $I_0(\cdot)$ is the modified Bessel function of the first kind.

According to (\ref{eq:pout}), $J_1(\zeta)$ is expressed as
\begin{align}
J_1(\zeta)
&=\int_0^{\lambda_1}
\frac{\exp\left(\frac{-y_1}{(1-\rho^2)\sigma_{h_0}^2}\right)}{(1-\rho^2)\sigma_{h_0}^2\sigma_{h_1}^2}
\int_0^{-\frac{\lambda_1 y_1}{\lambda_2}+\lambda_1}
\exp\left(\frac{-y_2}{(1-\rho^2)\sigma_{h_1}^2}\right)
I_0\left(\frac{2\rho\sqrt{y_1y_2}}{\sigma_{h_0}\sigma_{h_1}(1-\rho^2)}\right)\mathrm{d}y_2\mathrm{d}y_1
\nonumber\\
&=\sum_{m=0}^\infty\frac{\rho^{2m}}{(m!)^2(1-\rho^2)^m\sigma_{h_0}^{2(m+1)}}
\int_0^{\lambda_1}y_1^m\exp\left(\frac{-y_1}{(1-\rho^2)\sigma_{h_0}^2}\right)
\gamma\left(m+1,\frac{-\frac{\lambda_1 y_1}{\lambda_2}+\lambda_1}{(1-\rho^2)\sigma_{h_1}^2}\right)\mathrm{d}y_1
\end{align}
where we use the series representation of $I_0(z)$ \cite{refs:homfwfgamt}
\begin{align}
I_0(z) = \sum_{m=0}^\infty \frac{1}{(m!)^2} \left(\frac{z}{2}\right)^{2m}.
\end{align}

As the lower incomplete gamma function has the special case that \cite{refs:toisap}
\begin{align}
\gamma(m+1,x)=m!\left[1-\mathrm{e}^{-x}\left(\sum_{n=0}^m\frac{x^n}{n!}\right)\right],
\end{align}
using the binomial theorem, we have
\begin{align}
J_1(\zeta)
&=\sum_{m=0}^\infty\frac{\rho^{2m}}{m!(1-\rho^2)^m\sigma_{h_0}^{2(m+1)}}
\left[\int_0^{\lambda_1}y_1^m\exp\left(\frac{-y_1}{(1-\rho^2)\sigma_{h_0}^2}\right)\right.\mathrm{d}y_1-\exp\left(\frac{-\lambda_1}{(1-\rho^2)\sigma_{h_1}^2}\right)
\nonumber\\
&~~~~~~~\sum_{n=0}^m\sum_{k=0}^n
\frac{\binom{n}{k}\left(-\frac{\lambda_1}{\lambda_2}\right)^k\lambda_1^{n-k}}{n!(1-\rho^2)^n\sigma_{h_1}^{2n}}
\left.\int_0^{\lambda_1}y_1^{m+k}
\exp\left(\frac{-\lambda_1\sigma_{h_0}^2+\lambda_2\sigma_{h_1}^2}{(1-\rho^2)\lambda_2\sigma_{h_0}^2\sigma_{h_1}^2}y_1\right)
\mathrm{d}y_1\right]
\nonumber\\
&=\sum_{m=0}^\infty\frac{\rho^{2m}(1-\rho^2)}{m!}
\gamma\left(m+1,\frac{\lambda_1}{(1-\rho^2)\sigma_{h_0}^2}\right)-
\exp\left(\frac{-\lambda_1}{(1-\rho^2)\sigma_{h_1}^2}\right)\sum_{m=0}^\infty\sum_{n=0}^m\sum_{k=0}^n
\nonumber\\
&\frac{\binom{n}{k}(-1)^k\rho^{2m}(1-\rho^2)^{k-n+1}\lambda_1^n\lambda_2^{m+1}\sigma_{h_0}^{2k}\sigma_{h_1}^{2(m+k-n+1)}}
{m!n!\left(\lambda_2\sigma_{h_1}^2-\lambda_1\sigma_{h_0}^2\right)^{m+k+1}}
\gamma\left(m+k+1,\frac{\lambda_1\lambda_2\sigma_{h_1}^2-\lambda_1^2\sigma_{h_0}^2}{(1-\rho^2)\lambda_2\sigma_{h_0}^2\sigma_{h_1}^2}\right).
\end{align}
Similarly, we can obtain the second integration $J_2(\zeta)$ as
\begin{align}\label{eq:J2}
J_2(\zeta)
&=\sum_{m=0}^\infty\frac{\rho^{2m}}{(m!)^2(1-\rho^2)^m\sigma_{h_0}^{2(m+1)}}
\int_{\lambda_1}^\infty
y_1^m\exp\left(\frac{-y_1}{(1-\rho^2)\sigma_{h_0}^2}\right)
\nonumber\\
&~~~\left[\Gamma\left(m+1,\frac{\lambda_2-\frac{\lambda_2 y_1}{\lambda_1}}{(1-\rho^2)\sigma_{h_1}^2}\right)-
\Gamma\left(m+1,\frac{\lambda_1-\frac{\lambda_1 y_1}{\lambda_2}}{(1-\rho^2)\sigma_{h_1}^2}\right)\right]
\mathrm{d}y_1\triangleq J_{21}(\zeta)-J_{22}(\zeta).
\end{align}

Take the computation of the first part in (\ref{eq:J2}) as example, we have
\begin{align}
J_{21}(\zeta)& =\sum_{m=0}^\infty\frac{\rho^{2m}\exp\left(\frac{-\lambda_2}{(1-\rho^2)\sigma_{h_1}^2}\right)}{m!(1-\rho^2)^m\sigma_{h_0}^{2(m+1)}}
\int_{\lambda_1}^\infty y_1^m
\exp\left(\frac{\lambda_2\sigma_{h_0}^2-\lambda_1\sigma_{h_1}^2}{(1-\rho^2)\lambda_1\sigma_{h_0}^2\sigma_{h_1}^2}y_1\right)
\sum_{n=0}^m\frac{\left(\lambda_2-\frac{\lambda_2 y_1}{\lambda_1}\right)^n}{n!(1-\rho^2)^n\sigma_{h_1}^{2n}}\mathrm{d}y_1
\nonumber\\
&=\sum_{m=0}^\infty\frac{\rho^{2m}\exp\left(\frac{-\lambda_2}{(1-\rho^2)\sigma_{h_1}^2}\right)}{m!(1-\rho^2)^m\sigma_{h_0}^{2(m+1)}}
\sum_{n=0}^m\sum_{k=0}^n\frac{\binom{n}{k}(-\lambda_1)^{-k}\lambda_2^n}{n!(1-\rho^2)^n\sigma_{h_1}^{2n}}
\int_{\lambda_1}^\infty
y_1^{m+k}\exp\left[\frac{(\lambda_2\sigma_{h_0}^2-\lambda_1\sigma_{h_1}^2) y_1}{(1-\rho^2)\lambda_1\sigma_{h_0}^2\sigma_{h_1}^2}\right]
\mathrm{d}y_1
\nonumber\\
&=\exp\left(\frac{-\lambda_2}{(1-\rho^2)\sigma_{h_1}^2}\right)\sum_{m=0}^\infty\sum_{n=0}^m\sum_{k=0}^n
\frac{\binom{n}{k}(-1)^k\rho^{2m}(1-\rho^2)^{k-n+1}\lambda_1^{m+1}\lambda_2^n\sigma_{h_0}^{2k}\sigma_{h_1}^{2(m+k-n+1)}}
{m!n!(\lambda_1\sigma_{h_1}^2-\lambda_2\sigma_{h_0}^2)^{m+k+1}}
\nonumber\\
&~~~~\Gamma\left(m+k+1,\frac{\lambda_1\sigma_{h_1}^2-\lambda_2\sigma_{h_0}^2}{(1-\rho^2)\sigma_{h_0}^2\sigma_{h_1}^2}\right),
\end{align}
where the upper incomplete gamma function also has a special case that
\begin{align}
\Gamma(m+1,x)=m!\mathrm{e}^{-x}\left(\sum_{n=0}^m\frac{x^n}{n!}\right).
\end{align}

Therefore, (\ref{eq:outage}) can be obtained from $J_1(\zeta)+J_{21}(\zeta)-J_{22}(\zeta)$ with the relationship that $\Gamma(m+1,x)=m!-\gamma(m+1,x)$.

\section{ PDF of $X=\frac{|h_0|^2}{|h_1|^2}$}\label{sec:appendix1}

With the PDF definition of the ratio of two random variables \cite{refs:prvasp}, the PDF of $X$ can be obtained from
\begin{align}\label{eq:fx}
f_X(x)
&=\int_0^\infty y f_{|h_0|^2,|h_1|^2}(xy,y)\mathrm{d}y
 =\int_0^\infty C_1 y \mathrm{e}^{-C_2(x+\rho)y}I_0\left(C_3\sqrt{x}y\right)\mathrm{d}y  \nonumber\\
&=\sum_{m=0}^\infty \frac{C_1 C_3^{2m}x^m}{4^m(m!)^2}\int_0^\infty y^{2m+1}\mathrm{e}^{-C_2(x+\rho)y}\mathrm{d}y
=\sum_{m=0}^\infty\frac{C_{4m}x^m}{(1+x/\rho)^{2m+2}},~~~~~x\geq0
\end{align}
where
\begin{align}
C_1&=\frac{1}{(1-\rho^2)\sigma_{h_0}^2\sigma_{h_1}^2},~~~~~~~~~~~~~~~~
C_2 =\frac{1}{(1-\rho^2)\sigma_{h_0}^2},\\
C_3 &=\frac{2\rho}{\sigma_{h_0}\sigma_{h_1}(1-\rho^2)},~~~~~~~~~~~~~~~~
C_{4m}=\frac{(1-\rho^2)\rho^{m-1}(2m+1)!}{(m!)^2},
\end{align}
Thus the CDF of $X$ is given by
\begin{align}\label{eq:fx}
F_X(x)=\int_0^x f_X(y)\mathrm{d}y=\sum_{m=0}^\infty\frac{C_{4m}x^{m+1}}{m+1}{}_2F_1\left(2m+2,m+1;m+2,-\frac{x}{\rho}\right).
\end{align}

\linespread{1.2}

\end{document}